\documentclass[a4paper]{article}
\usepackage{amssymb,amsmath,amsthm,verbatim}
\usepackage[dvips,final]{graphicx}
\usepackage{psfrag}
\usepackage{natbib}
\usepackage{hyperref}

\newcommand{\R}{\mathbb{R}}
\newcommand{\Prob}{\mathbb{P}}
\renewcommand{\Pr}{\Prob}
\renewcommand{\bf}{\bfseries}

\newcommand{\E}{\mathbb{E}}

\newcommand{\Fc}{\mathcal{F}}

\newcommand{\N}{\mathbb{N}}

\newcommand{\eps}{\varepsilon}

\newtheorem{theorem}{Theorem}[section]
\newtheorem{prop}[theorem]{Proposition}
\newtheorem{lemma}[theorem]{Lemma}
\newtheorem{cor}[theorem]{Corollary}

\theoremstyle{definition}

\theoremstyle{remark}
\newtheorem{remark}[theorem]{Remark}
\newtheorem{definition}[theorem]{Definition}

\newcommand{\indic}[1]{\boldsymbol{1}_{\{\ensuremath{#1}\}}}

\newcommand{\ol}[1]{\overline{#1}}

\newcommand{\as}{a.s.}

\newcommand{\eg}{e.g.}

  {\end{list}}
\newcounter{easyenum}

  {\end{list}}
\newcounter{mediumenum}

  {\end{list}}
\newcounter{hardenum}

\newcommand{\gp}{\gamma_{+}(\lambda)}
\newcommand{\gm}{\gamma_{-}(\lambda)}

\newcommand{\K}{\mathbb{K}} %
\newcommand{\pathsp}{\mathfrak{P}}
\newcommand{\Pc}{\mathcal{P}}
\newcommand{\Lin}{\mathrm{Lin}}
\newcommand{\stset}{\mathfrak{T}}

\newcommand{\ub}{\ensuremath{\overline{b}}} %
\newcommand{\lb}{\ensuremath{\underline{b}}} %
\newcommand{\uK}{\ensuremath{\overline{K}}} %
\newcommand{\lK}{\ensuremath{\underline{K}}} %
\newcommand{\dnt}{\mathbf{1}_{\overline{S}_T< \ub,\, \underline{S}_T>
 \lb}} %
\newcommand{\uh}{\overline{H}} %
\newcommand{\lh}{\underline{H}}%
\newcommand{\sS}{\overline{S}} %
\newcommand{\iS}{\underline{S}}%
\newcommand{\sB}{\overline{B}} %
\newcommand{\iB}{\underline{B}}%
\newcommand{\uB}{\sB} %
\newcommand{\lB}{\iB}%
\newcommand{\sM}{\overline{M}} %
\newcommand{\iM}{\underline{M}}%
\newcommand{\sD}{\overline{D}} %
\newcommand{\iD}{\underline{D}}%
\renewcommand{\gp}{\gamma^+_{\mu}}%
\renewcommand{\gm}{\gamma^{-}_{\mu}}%
\def\ovl{\overline}

\title{Robust pricing and hedging of double no-touch options}
\author{Alexander~M.~G.~Cox\thanks{e-mail:
      \texttt{A.M.G.Cox@bath.ac.uk};
      web: \texttt{www.maths.bath.ac.uk/$\sim$mapamgc/}}\\
      Dept.\ of Mathematical Sciences, 
      University of Bath\\
      Bath BA2 7AY, UK
      \and
Jan Ob\l\'oj\thanks{e-mail: \texttt{obloj@maths.ox.ac.uk}; web: \texttt{www.maths.ox.ac.uk/contacts/obloj}. 
Research supported by a Marie Curie Intra-European Fellowship at Imperial College London within the $6^{th}$ European Community Framework Programme.}\\
     Oxford-Man Institute of Quantitative Finance \emph{and}\\
     Mathematical Institute, University of Oxford\\
      24-29 St Giles, Oxford OX1 3LB, UK
} 

\begin{document}
\maketitle
\begin{abstract}
Double no-touch options, contracts which pay out a fixed amount provided an underlying asset remains within a given interval, are commonly traded, particularly in FX markets. In this work, we establish model-free bounds on the price of these options based on the prices of more liquidly traded options (call and digital call options). Key steps are the construction of super- and sub-hedging strategies to establish the bounds, and the use of Skorokhod embedding techniques to show the bounds are the best possible.

In addition to establishing rigorous bounds, we consider carefully what is meant by arbitrage in settings where there is no {\it a priori} known probability measure. We discuss two natural extensions of the notion of arbitrage, weak arbitrage and weak free lunch with vanishing risk, which are needed to establish equivalence between the lack of arbitrage and the existence of a market model.
\end{abstract}

\section{Introduction}
It is classical in the Mathematical Finance literature to begin by
assuming the existence of a filtered probability space $(\Omega, \Fc,
(\Fc_t)_{t \ge 0}, \Pr)$ on which an underlying price process is
defined. In this work we do not assume any given model. Instead we are
given the observed prices of vanilla options and our aim is to derive
information concerning the arbitrage-free price of an exotic option, while assuming
as little as possible about the underlying asset's behaviour.

More precisely, our starting point is the following question: suppose we
know the call prices on a fixed underlying at a given maturity date,
what can we deduce about the prices of a double no-touch option,
written on the same underlying and settled at the same maturity as the
call options? A double no-touch option is a contract which pays a fixed amount at maturity (which we will assume always
to be 1 unit), provided the asset remains (strictly) between two fixed
barriers. They appear most commonly in FX markets, and (in different
contexts to the one we consider) have been considered recently by
\eg{} \cite{CarrCrosby:08} and \cite{Mijatovic:08}.

The approach we take to the problem is based on the approach which was
initially established in \cite{Hobson:98b}, and later in different
settings in \cite{Brown:01b,CoxHobsonObloj:08,CoxObloj:08}. The
basic principle is to use constructions from the theory of Skorokhod
embeddings to identify extremal processes, which may then give
intuition to identify optimal super- and sub-hedges. A model based on
the extremal solution to the Skorokhod embedding allows one to deduce
that the price bounds implied by the hedges are tight. In the setting considered here, as we will show shortly, the
relevant constructions already exist in the Skorokhod embedding
literature (due to \cite{Perkins:86,Jacka:88,CoxHobson:06}). However the
hedging strategies have not been explicitly derived. One of the goals of this paper is to open up these results to the finance community.

A second aspect of our discussion concerns a careful consideration of
the technical framework in which our results are valid: we let the `market' determine a set of asset prices, and we assume that
these prices satisfy standard linearity assumptions. In particular, our
starting point is a linear operator on a set of functions from a path
space (our asset histories) to the real line (the payoff of the
option). Since there is no specified probability measure, a suitable
notion of arbitrage has to be introduced: the simplest arbitrage
concept here is that any non-negative payoff must be assigned a
non-negative price, which we call a `model-free arbitrage'. However,
as noted in \cite{DavisHobson:06}, this definition is insufficient to
exclude some undesirable cases. In \cite{DavisHobson:06}, this issue
was resolved by introducing the notion of `weak arbitrage,' and this
is a concept we also introduce, along with the notion of `weak free
lunch with vanishing risk.' Our main results are then along the following lines: if the stated prices satisfy the stronger
no-arbitrage condition, then there exists a market model, i.e.\ a probability space with a stock price process which is a martingale and such that the expectation agrees with the pricing operator. On the other hand, if we see prices which exhibit no model-free arbitrage, but which
admit a weaker arbitrage, then there is no market model, but we are restricted to a `boundary' between the prices for which there is a
martingale measure, and the prices for which there exists a model-free arbitrage. Interestingly, we only need one of the stronger types of
arbitrage depending on the call prices considered -- when we consider markets in which calls trade at all strikes, then the
weak free lunch with vanishing risk condition is needed. When we suppose only a finite number of strikes are traded, then the weak
arbitrage condition is required.

We note that there are a number of papers that have considered a
similar `operator' based approach, where certain prices are specified,
and an arbitrage concept introduced: \cite{BiaginiCont:07} suppose the
existence of a pricing operator satisfying a number of conditions,
which turn out to be sufficient to deduce the existence of a
probability measure (although their conditions mean that the pricing
operator has to be defined for all bounded payoffs); \cite{Cherny:07}
considers a similar setting, but with a stronger form of arbitrage,
from which a version of the Fundamental Theorem of Asset Pricing is
recovered. \cite{Cassese:08} also considers a similar setup, and is
able to connect notions of arbitrage to the existence of finitely
additive martingale measures, and under certain conditions, to the
existence of martingale measures.

The paper is organised as follows. We first carefully introduce our
setup and then in Section~\ref{sec:arb} we study when call prices (at
a finite or infinite number of strikes), possibly along with some
digital calls, are free of different types of arbitrage and when they
are compatible with a market model. Then in Section \ref{sec:hedges}
we construct sub- and super-hedges of a double no-touch option which
only use calls and puts, digital calls at the barriers and forward
transactions. We then combine these hedges with no arbitrage results
and in Section~\ref{sec:prices} we determine the range of
arbitrage-free (for different notions of arbitrage) prices of double
no-touch options given prices of calls and digital calls. Finally in
Section~\ref{sec:app} we discuss possible applications and present
some brief numerical simulations. Additional technical results about
weak arbitrage along with the proof of Theorem~\ref{thm:ubfinstrike}
are given in Appendix~\ref{ap:weakarb}.  Appendix~\ref{ap:mart}
contains some remarks about the joint law of the maximum and minimum
of a uniformly integrable martingale with a given terminal law, which
follow from the results in this paper.

\subsection{Market input and the modelling setup}
Our main assumptions concern the behaviour of the asset price
$(S_t)_{t \ge 0}$. We will assume that the asset has zero cost of
carry: this can come about in a number of ways --- $S_t$ might be a
forward price, the asset might be paying dividends continuously, at a
rate equal to the prevailing interest rate, or the underlying could be
the exchange rate between two economies with similar interest
rates. In addition, through most of the paper, we assume that the
paths of $S_t$ are continuous. In principle, these are really the only
assumptions that we need on the asset, combined with some assumptions
on the behaviour of the market (that the asset and certain derivative
products -- calls and forwards -- are traded without transaction costs
and at prices which are themselves free of arbitrage). Specifically,
we do not need to assume that the price process is a semi-martingale,
or that there is any probability space or measure given.  The
statement of the results will concern the existence of an arbitrage
for {\it any} path which satisfies the above conditions, or
alternatively the existence of an arbitrage free model which satisfies
the above conditions.

More formally, we let $(S_t:t\leq T)$ be an element of $\pathsp$ --
the space of possible paths of the stock price process. Let us assume
that $\pathsp=C([0,T]; S_0)$ is the space of \emph{continuous
non-negative} functions on $[0,T]$ with a fixed initial value
$S_0>0$, we shall discuss extensions to discontinuous setups later on.
We suppose there are number of traded assets, which we define as
real-valued functions on $\pathsp$, priced by the market. The simplest
transactions which we price are `constants,' where we assume that the
constant payoff $F$ can be purchased at initial cost $F$ (recall that
we are assuming zero cost of carry). Then we assume that calls with
strikes $K\in\K$, with payoffs $(S_T-K)^+$, are traded at respective
prices $C(K)$, where $\K$ is a set of strikes $\K\subset \R_+$.
Finally, we also assume that forward transactions have zero cost. A
forward transaction is one where at some time $\rho$ parties exchange
the current value of the stock price $S_\rho$ against the terminal
value $S_T$. More precisely, let $(\mathcal{F}^n_t)_{t\leq T}$ be the
natural filtration of the co-ordinate process on $C([0,T])$ and
consider a class $\stset$ of stopping times $\rho$ relative to
$(\mathcal{F}^n_t)_{t\leq T}$. Then the forward transaction has payoff
$(S_T-S_\rho)\mathbf{1}_{\rho\leq T}$. In particular, we will consider here hitting times of levels $H_b:\pathsp\to [0,T]\cup \{\infty\}$
defined by $H_b=\inf\{t\leq T: S_t=b\}$, $b\geq 0$. As the paths are continuous these are indeed stopping times and we have
$H_b=\inf\{t\leq T: S_t\geq b\}$ for $b\geq S_0$, with similar expression for $b\leq S_0$, and $S_{H_b}=b$ whenever $H_b\leq T$.  We
let
\begin{equation}
\label{eq:Xsetofassets}
\mathcal{X} = \left\{F, (S_T-K)^+, (S_T-S_{\rho})\mathbf{1}_{\rho\leq
  T}:F\in \R,\ K\in \K,\ \rho\in \stset\right\}.
\end{equation}
Note that here $\mathcal{X}$ is a just a set of real-valued functions
on $\pathsp$. We denote $\Lin(\mathcal{X})$ the set of finite linear
combinations of elements of $\mathcal{X}$. \\
We assume the prices of elements of $\mathcal{X}$ are known in the
market, as discussed above, and portfolios of assets in $\mathcal{X}$
are priced linearly. More precisely, we suppose there exists a
\emph{pricing operator} $\mathcal{P}$ defined on $\Lin(\mathcal{X})$
which is linear, and satisfies the following rules:
\begin{eqnarray}
\mathcal{P} 1 & = & 1; \label{eq:pricop1}\\
\mathcal{P} (S_T - K)^+ & = & C(K), \quad \forall K \in \mathbb{K}; \label{eq:pricop2}\\
\mathcal{P} (S_T - S_\rho) \indic{\rho \le T} & = & 0,\quad \forall \rho\in\stset. \label{eq:pricop3}
\end{eqnarray}
Later in the paper we will consider examples of $\stset$ but for now it is arbitrary and we only assume $0\in \stset$. It then follows that $S_T=(S_T-S_0)\mathbf{1}_{0\leq T}+S_0$ is an element of $\Lin(\mathcal{X})$ and $\Pc S_T=S_0$.\footnote{We note
that in some financial markets, in particular in the presence of
bubbles, it may be sensible to assume that $\Pc S_T=\Pc C(0)\neq
S_0$ even when the cost of carry is zero, see \cite{CoxHobson:06}.}
Note also that $S_T=(S_T)^+$ so that we can always assume $0\in\K$ and
we have $C(0)=S_0$ by linearity of $\Pc$.\footnote{Alternatively, we
could have assumed $C(0)=S_0$ and then deduce from no arbitrage that
$(S_t)$ has non-negative paths.} We deduce that European puts are
also in $\Lin(\mathcal{X})$, we write $P(K)$ for their prices, and
note that the put-call parity follows from linearity:
\[
\mathcal{P} (K-S_T)^+ = \mathcal{P} (K - S_T + (S_T-K)^+) = K - S_0
+C(K).
\]
In later applications, we will also at times wish to suppose that
$\mathcal{X}$ is a larger set and in particular that $\mathcal{P}$
also prices digital call options at certain strikes.

The above setup is rather general and we relate it now to more
`classical models'.
\begin{definition}\label{def:model}
A \emph{model} is a probability space $(\Omega, \mathcal{F}, \Pr)$
with filtration $(\mathcal{F}_t)_{t\leq T}$ satisfying the usual
hypothesis and an adapted stochastic process $(S_t)$ with paths in
$\pathsp$. \\
A model is called a $(\Pc,\mathcal{X})$-\emph{market model} if
$(S_t)$ is a $\Pr$-martingale and
\begin{equation}\label{eq:marketmodel}
\forall X\in \mathcal{X}\quad \E[X((S_t:t\leq T))]=\mathcal{P}X,
\end{equation}
where we implicitly assume that the LHS is well defined. 
\end{definition}
Note that for a $(\Pc,\mathcal{X})$-market model \eqref{eq:marketmodel} holds for all $X\in\Lin(\mathcal{X})$ by linearity of $\E$ and $\Pc$.
The notion of market model is relative to the market input, i.e.\ the set of assets $\mathcal{X}$ and their prices $\Pc(X)$, $X\in
\mathcal{X}$. However when $(\Pc,\mathcal{X})$ are clear from the
context we simply say that there exists a market model.

Our aim is to understand the possible extensions of $\mathcal{P}$ to
$\Lin(\mathcal{X} \cup \{Y\})$, where $Y$ is the payoff of an
additional asset, in particular of a barrier option. Specifically, we
are interested in whether there is a linear extension which preserves
the no-arbitrage property:
\begin{definition} \label{def:Pnoarb} We say a pricing operator
$\mathcal{P}$ admits {\it no model-free arbitrage} on $\mathcal{X}$
if
\begin{equation} \label{eq:Pnoarb}
 \forall X \in \Lin(\mathcal{X}):
 X \ge 0 \implies \mathcal{P}X \ge 0.
\end{equation}
\end{definition}
Naturally, whenever there exists a market model then we can extend
$\mathcal{P}$ using \eqref{eq:marketmodel} to all payoffs $X$ for
which $\E |X(S_t:t\leq T)|<\infty$.  In analogy with the Fundamental
Theorem of Asset Pricing, we would expect the following dichotomy:
either there is no extension which preserves the no-arbitrage
property, or else there is a market model and hence a natural
extension for $\mathcal{P}$.  To some extent, this is the behaviour we
will see, however model-free arbitrage is too weak to grant this
dichotomy. One of the features of this paper is the introduction of
\emph{weak free lunch with vanishing risk} criterion (cf.~Definition
\ref{def:wflvr}) which is then applied together with \emph{weak
arbitrage} of Davis and Hobson \cite{DavisHobson:06}.

{\bf Notation:}
The minimum and maximum of two numbers are denoted $a\land
b=\min\{a,b\}$ and $a\lor b=\max\{a,b\}$.  The running maximum and
minimum of the price process are denoted respectively
$\overline{S}_t=\sup_{u\leq t}S_u$ and $\underline{S}_t=\inf_{u\leq
t}S_u$. We are interested in derivatives with digital payoff
conditional on the price process staying in a given range. Such an
option is often called a \emph{double no-touch} option or a
\emph{range option} and has payoff $\dnt$.  It is often convenient to
express events involving the running maximum and minimum in terms of
the first hitting times $H_x=\inf\{t: S_t=x\}$, $x\geq 0$. As an
example, note that when the asset is assumed to be continuous, we have
$\dnt=\mathbf{1}_{H_{\ub}\wedge H_{\lb}> T}$.

\section{Arbitrage-free prices of call and digital options}\label{sec:arb}
Before we consider extensions of $\Pc$ beyond $\mathcal{X}$ we need to
understand the necessary and sufficient conditions on the market
prices which guarantee that $\Pc$ does not admit model-free arbitrage
on $\mathcal{X}$. This and related questions have been considered a
number of times in the literature, e.g.~Hobson \cite{Hobson:98b}, Carr
and Madan \cite{CarrMadan:05}, Davis and Hobson \cite{DavisHobson:06},
however never in the full generality of our setup, and there remained
some open issues which we resolve below.


It turns out the constraints on $C(\cdot)$ resulting from the
condition of no model-free arbitrage of Definition~\ref{def:Pnoarb}
are not sufficient to guarantee existence of a market model. We will
give examples of this below both when $\K=\R_+$ and when $\K$ is
finite (the latter coming from Davis and Hobson
\cite{DavisHobson:06}). This phenomena motivates stronger notions of
no-arbitrage.
\begin{definition}\label{def:wflvr}
We say that the pricing operator $\mathcal{P}$ admits a {\it weak
 free lunch with vanishing risk} (WFLVR) on $\mathcal{X}$ if there
exist $(X_n)_{n \in \N}, Z \in \Lin(\mathcal{X})$ such that $X_n \to
X$ pointwise on $\pathsp$, $X_n \ge Z$, $X\ge 0$ and $\lim_n
\mathcal{P} X_n <0$.
\end{definition}
Note that if $\Pc$ admits a model-free arbitrage on
$\mathcal{X}$ then it also admits a WFLVR on $\mathcal{X}$. No WFLVR
is a stronger condition as it also tells us about the behaviour of
$\Pc$ on (a certain) closure of $\mathcal{X}$. It is naturally a weak
analogue of the NFLVR condition of Delbaen and Schachermayer
\cite{DelbaenSchachermayer:04}.

This new notion proves to be sufficiently strong to guarantee
existence of a market model.
\begin{prop}\label{prop:calls_wflvr}
Assume $\K=\R_+$. Then $\Pc$ admits no WFLVR on $\mathcal{X}$ if and
only if there exists a $(\Pc,\mathcal{X})$-market model, which
happens if and only if
\begin{eqnarray}
 C(\cdot)\textrm{ is a non-negative, convex, decreasing function, }C(0)=S_0,\ C'_+(0)\geq -1,\label{eq:Cbasic}\\
 C(K)\to 0\textrm{ as }K\to\infty.\label{eq:Ccnv}
\end{eqnarray}
In comparison, $\Pc$ admits no model-free arbitrage on $\mathcal{X}$
if and only if \eqref{eq:Cbasic} holds. In consequence, when
\eqref{eq:Cbasic} holds but \eqref{eq:Ccnv} fails $\Pc$ admits no
model-free arbitrage but a market model does not exist.
\end{prop}
\begin{proof}
That absence of a model-free arbitrage implies \eqref{eq:Cbasic} is
straightforward and classical. Note that since $C(\cdot)$ is convex
$C'_+(0)=C'(0+)$ is well defined. Let
$\alpha:=\lim_{K\to\infty}C(K)$ which is well defined by
\eqref{eq:Cbasic} with $\alpha\geq 0$. If $\alpha>0$ then
$X_n=-(S_T-n)^+$ is a WFLVR since $X_n\to 0$ pointwise as
$n\to\infty$ and $\Pc X_n=-C(n)\to-\alpha<0$. We conclude that no
WFLVR implies \eqref{eq:Cbasic}--\eqref{eq:Ccnv}. But then we may
define a measure $\mu$ on $\R_+$ via
\begin{equation} \label{eq:mudef}
 \mu([0,K]) = 1+C_+'(K),
\end{equation}
which is a probability measure with $\mu([K,\infty))=-C_-'(K)$ and 
\[
\int x \mu(dx) = \int \mu((x,\infty)) \, dx = - \int C_+'(x) \, dx =C(0) - C(\infty) = S_0.
\]
In fact, \eqref{eq:mudef} is the well known relation between the
risk neutral distribution of the stock price and the call prices due
to Breeden and Litzenberger \cite{BreedenLitzenberger:78}. Let
$(B_t)$ be a Brownian motion, $B_0=S_0$, relative to its natural
filtration on some probability space $(\Omega,\mathcal{F},\Pr)$ and
$\tau$ be a solution to the Skorokhod embedding problem for $\mu$,
i.e.\ $\tau$ is a stopping time such that $B_\tau$ has law $\mu$ and
$(B_{t\wedge \tau})$ is a uniformly integrable martingale. Then
$S_t:=B_{\frac{t}{T-t}\wedge\tau}$ is a continuous martingale with
$T$-distribution given by $\mu$. Hence it is a market model as $\E (S_T-K)^+=C(K)$ by \eqref{eq:mudef} and $\E (S_T-S_\rho)\mathbf{1}_{\rho<T}=0$ for all stopping times, in particular for $\rho\in\stset$, since $(S_t)$ is a martingale.
Finally, whenever a market model exists then clearly we have no WFLVR since $\Pc$ is the expectation. \\
It remains to argue that \eqref{eq:Cbasic} alone implies that $\Pc$
admits no model-free arbitrage. Suppose to the contrary that there
exists $X\in\Lin(\mathcal{X})$ such that $X \ge 0$ and
$\Pc(X)<\epsilon<0$. As $X$ is a finite linear combination of
elements of $\mathcal{X}$, we let $\ovl K$ be the largest among the
strikes of call options present in $X$. Then, for any $\delta>0$,
there exists a function $C_{\delta}$ satisfying
\eqref{eq:Cbasic}--\eqref{eq:Ccnv} and such that $C(K)\geq
C_\delta(K)\geq C(K)-\delta$, $K\leq \ovl K$. More precisely, if
$C_+'(\ovl K)>0$ then $C$ is strictly decreasing on $[0,\ovl K]$ and
we may in fact take $C_\delta=C$ on $[0,\ovl K]$. Otherwise $C$ is
constant on some interval $[K_0,\infty)$ but then either it is zero
or we can clearly construct $C_\delta$ which approximates it
arbitrarily closely on $[0,\ovl K]$. By the arguments above, the pricing operator
$\Pc_\delta$ corresponding to prices $C_\delta$ satisfies no WFLVR
and hence no model-free arbitrage, so $\Pc_\delta X \ge 0$. However, we can take $\delta$ small enough so that $|\Pc
X-\Pc_\delta X|<\epsilon/2$ which gives the desired contradiction.
\end{proof}
We now turn to the case where $\K$ is a finite set. The no WFLVR
condition does not appear to be helpful here, and we need to use a different
notion of no-arbitrage due to Davis and Hobson \cite{DavisHobson:06}.
\begin{definition}
We say that a pricing operator $\mathcal{P}$ admits a \emph{weak
 arbitrage} (WA) on $\mathcal{X}$ if, for any model $\Pr$, there exists an $X\in\Lin(\mathcal{X})$ such that
$\Pc X \le 0$ but $\Pr(X((S_t:t\leq T)) \ge 0)=1$ and
$\Pr(X((S_t:t\leq T))>0) >0$.
\end{definition}
Note that no WA implies no model-free
arbitrage. Indeed, let $X$ be a model-free arbitrage so that $X\geq 0$
and $\Pc(X)<\epsilon<0$. Then $Y=X+\epsilon/2$ is a WA since $Y>0$ on
all paths in $\pathsp$ and $\Pc(Y)<0$. In addition, the existence of a
market model clearly excludes weak arbitrage. Strictly speaking, our
definition of a weak arbitrage differs from \cite{DavisHobson:06},
since we include model-free arbitrages in the set of weak arbitrages.

\begin{prop}[Davis and Hobson
\cite{DavisHobson:06}] \label{prop:calls_WA}
Assume $\K\subset \R_+$ is a finite set. Then $\Pc$ admits no WA on
$\mathcal{X}$ if and only if there exists a
$(\Pc,\mathcal{X})$-market model, which happens if and only if
$C(K)$, $K\in \K$, may be extended to a function $C$ on $\R_+$
satisfying \eqref{eq:Cbasic}--\eqref{eq:Ccnv}. \\
Furthermore, $\Pc$ may admit no model-free arbitrage but admit a WA.
\end{prop}
\begin{proof}
Clearly if $C(K)$, $K\in\K$ may be extended to a function $C$ on
$\R_+$ satisfying \eqref{eq:Cbasic}--\eqref{eq:Ccnv} then from
Proposition \ref{prop:calls_wflvr} there exists a market model,
$\Pc$ is the expectation and hence there is no WA. If $C(K)$ is not
convex, positive or decreasing on $\K$ then a model-free arbitrage
can be constructed easily. It remains to see what happens if
$C(K_1)=C(K_2)=\alpha>0$ with $K_1<K_2$.  This alone does not entail
a model-free arbitrage as observed in the proof of Proposition
\ref{prop:calls_wflvr}. However a WA can be constructed as follows:
\begin{equation*}
\left\{
 \begin{array}{rclcl}
   X & = & (S_T - K_1)^+ - (S_T - K_2)^+ & \quad &\ \mbox{ if }
   \Pr(S_T > K_1) >0,\textrm{ else}\\
   X & = & \alpha - (S_T - K_1)^+ & \quad &\ \mbox{ if }
   \Pr(S_T > K_1) =0.
 \end{array}
\right.
\end{equation*}
\end{proof}

We now enlarge the set of assets to include digital calls. More
precisely let $0<\lb<\ub$ and consider 
\begin{equation}\label{eq:XDdef}
\mathcal{X}_D=\mathcal{X}\cup\{\indic{S_T >\lb},\indic{S_T \geq\ub}\}
\end{equation}
setting $\Pc$ on the new
assets to be equal to their (given) market prices: $\Pc\indic{S_T>
\lb}=\iD(\lb)$ and $\Pc\indic{S_T \geq \ub}=\sD(\ub)$ and imposing
linearity on $\Lin(\mathcal{X}_D)$.
\begin{prop}\label{prop:wflvr}
Assume $\K=\R_+$. Then $\Pc$ admits no WFLVR on $\mathcal{X}_D$ if
and only if there exists a $(\Pc,\mathcal{X}_D)$-market model, which
happens if and only if \eqref{eq:Cbasic}--\eqref{eq:Ccnv} hold and
\begin{equation} \label{eq:D_basic}
 \iD(\lb) = -C_+'(\lb) \quad\textrm{and}\quad \sD(\ub)  = - C_-'(\ub).
\end{equation}
\end{prop}
\begin{proof}
Obviously existence of a market model implies no WFLVR. The
convergences (pointwise in $\pathsp$), as $\epsilon\to 0$,
\begin{eqnarray*}
\frac{1}{\eps} \left[ (S_T-(K-\eps))^+ - (S_T - K)^+\right] & \to &
\indic{S_T \ge K},\\
\frac{1}{\eps} \left[ (S_T-K)^+ - (S_T - (K+\eps))^+\right] & \to &
\indic{S_T > K},
\end{eqnarray*}
readily entail that no WFLVR implies \eqref{eq:D_basic} hold and from
Proposition \ref{prop:calls_wflvr} it also implies
\eqref{eq:Cbasic}--\eqref{eq:Ccnv}. Finally, if
\eqref{eq:Cbasic}--\eqref{eq:Ccnv} hold we can can consider a
$(\Pc,\mathcal{X})$-market model by Proposition
\ref{prop:calls_wflvr}. We see that $\E \indic{S_T \geq
\ub}=\mu([\ub,\infty))=-C_-'(\ub)$ and $\E \indic{S_T >
\lb}=\mu((\lb,\infty))=-C_+'(\lb)$ from \eqref{eq:mudef}, and hence
\eqref{eq:D_basic} guarantees that the model matches $\Pc$ on
$\mathcal{X}_D$, i.e.\ we have a market model.
\end{proof}
Finally, we have an analogous proposition in the case where finitely
many strikes are traded.
\begin{prop}\label{prop:wa}
Assume $\K\subset\R_+$ is a finite set and $\lb,\ub\in \K$. Then
$\Pc$ admits no WA on $\mathcal{X}_D$ if and only if there exists a ($\Pc,\mathcal{X}_D)$-market model, which happens if and only if
$C(K)$, $K\in \K$, may be extended to a function $C$ on $\R_+$
satisfying \eqref{eq:Cbasic}--\eqref{eq:Ccnv} and \eqref{eq:D_basic}.
\end{prop}
The Proposition follows from Lemma~\ref{lem:wkarb} given in the
appendix, which also details the case $\lb,\ub\notin\K$.
\section{Robust hedging strategies}\label{sec:hedges}
We turn now to robust hedging of double no-touch options. We fix
$\lb<S_0<\ub$ and consider the derivative paying $\dnt$. Our aim is to
devise simple super- and sub- hedging strategies (inequalities) using
the assets in $\mathcal{X}_D$. If successful, such inequalities will
instantly yield bounds on $\Pc (\dnt)$ under the assumption of no
model-free arbitrage.
\subsection{Superhedges}
We devise now simple a.s.\ inequalities of the form
\begin{equation}\label{eq:super_gen}
\dnt \leq N_T + g(S_T),
\end{equation}
where $(N_t:t\leq T)$, when considered in a market model, is a
martingale. That is, we want $(N_t)$ to have a simple interpretation
in terms of a trading strategy and further $(N_t)$ should ideally only
involve assets from $\mathcal{X}_D$. A natural candidate for $(N_t)$
is a sum of terms of the type $\beta (S_t-z)\mathbf{1}_{\sS_t\geq z}$,
which is a purchase of $\beta$ forwards when the stock price reaches
the level $z$. We note also that in a market model, it is a simple
example of an Az\'ema-Yor martingale, that is of a martingale which is
of the type $H(S_t,\sS_t)$ for some function $H$ (see Ob\l\'oj
\cite{Obloj:06}).

We give three instances of \eqref{eq:super_gen}. They correspond in
fact to the three types of behaviour of the `extremal market model'
which maximises the price of the double no-touch option, which we will
see below in the proof of Lemma \ref{lem:perkins}.
\begin{enumerate}
\item We take $N_t\equiv 0$ and $g(S_T)=\mathbf{1}_{S_T\in
 (\lb,\ub)}=:\uh^I$. The superhedge is static and consists simply
of buying a digital options paying $1$ when $\lb<S_T<\ub$.

\psfrag{lb}{$\lb$}
\psfrag{K}{$K$}
\psfrag{Initial}{\small Portfolio for $t < H_{\lb}$}
\psfrag{Final}{\small Portfolio for $t \ge H_{\lb}$}
\begin{figure}[t]
 \includegraphics[width=\textwidth]{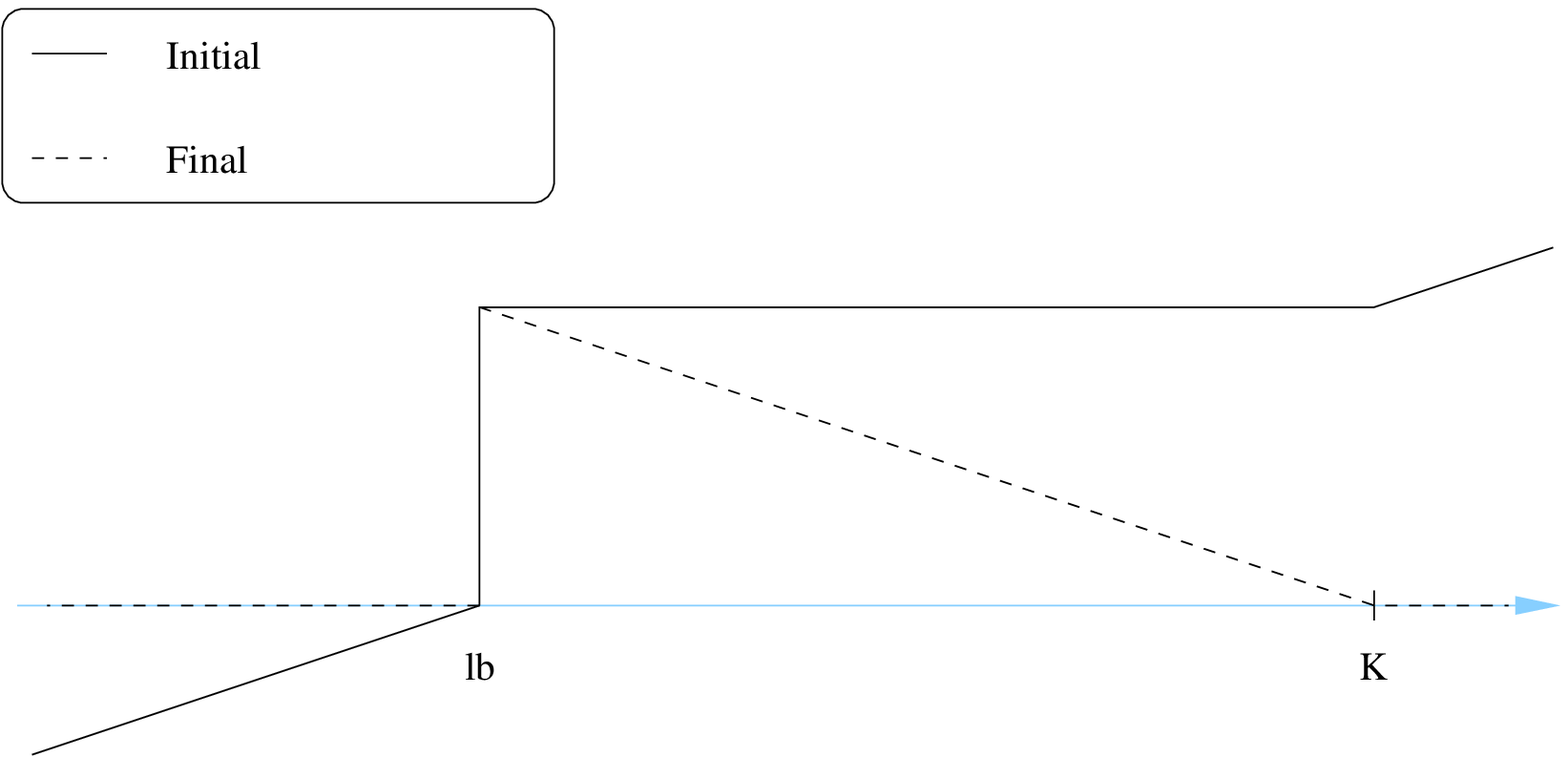}
 \caption{\label{fig:uH2} The value of the portfolio $\uh^{II}$ as a
   function of the asset price.}
\end{figure}

\item In this case we superhedge the double no-touch option as if it was simply
a barrier option paying $\mathbf{1}_{\iS_T>\lb}$ and we adapt the
superhedge from Brown, Hobson and Rogers \cite{Brown:01b}. More
precisely we have
\begin{equation}\label{eq:uH2}
 \dnt\leq \mathbf{1}_{S_T>\lb}- \frac{(\lb-S_T)^+}{K-\lb} +
 \frac{(S_T-K)^+}{K-\lb} - \frac{S_T-\lb}{K-\lb}
 \mathbf{1}_{\iS_T\leq \lb} =: \uh^{II}(K),
\end{equation}
where $K>\lb$ is an arbitrary strike. The portfolio $\uh^{II}(K)$ is
a combination of initially buying a digital option paying $1$ if
$S_T>\lb$, buying $\alpha=1/(K-\lb)$ calls with strike $K$ and
selling $\alpha$ puts with strike $\lb$. Upon reaching $\lb$ we then
sell $\alpha$ forward contracts. Note that $\alpha$ is chosen so
that our portfolio is worth zero everywhere except for $S_T\in
(\lb,K)$ after selling the forwards. This is represented graphically
in Figure~\ref{fig:uH2}.

\item We mirror the last case but now we superhedge a barrier option
paying $\mathbf{1}_{\sS_T<\ub}$.  We have
\begin{equation}\label{eq:uH3}
 \dnt\leq \mathbf{1}_{S_T<\ub}+\frac{(K-S_T)^+}{\ub-K} -
 \frac{(S_T-\ub)^+}{\ub-K} + \frac{S_T-\ub}{\ub-K}
 \mathbf{1}_{\sS_T\geq \ub} =:\uh^{III}(K),
\end{equation}
where $K<\ub$ is an arbitrary strike. In this case the portfolio
$\uh^{III}(K)$ consists of buying a digital option paying $1$ on
$\{S_T<\ub\}$ and $\alpha=1/(\ub-K)$ puts with strike $K$, selling
$\alpha$ calls with strike $\ub$ and buying $\alpha$ forwards when
the stock price reaches $\ub$. Similarly to \eqref{eq:uH2}, $\alpha$
is chosen so that $\uh^{III}(K)=0$ for $S_T\not\in (K,\ub)$ when
$\sS_T\geq \ub$, i.e.~when we have carried out the forward transaction.
\end{enumerate}

\subsection{Subhedges}

\psfrag{lb}{$\lb$}
\psfrag{K1}{$K_1$}
\psfrag{ub}{$\ub$}
\psfrag{K2}{$K_2$}
\psfrag{Initial}{\small Portfolio for $t < H_{\lb}\wedge H_{\ub}$}
\psfrag{Final1}{\small Portfolio for $H_{\lb} < t \wedge H_{\ub}$}
\psfrag{Final2}{\small Portfolio for $H_{\ub} < t \wedge H_{\lb}$}
\begin{figure}[t]
\includegraphics[width=\textwidth]{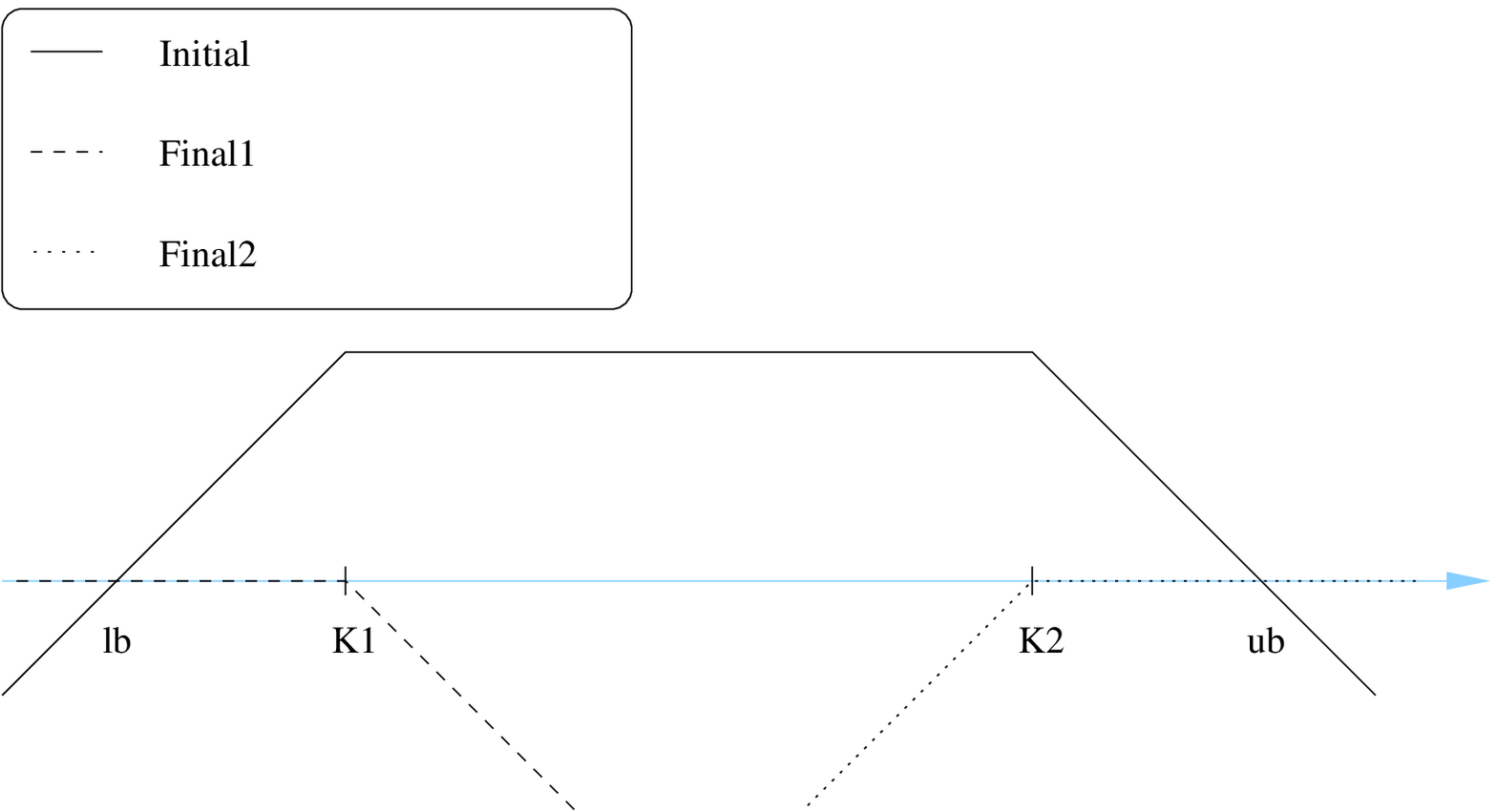}
\caption{\label{fig:lH2} The value of the portfolio $\lh^{II}$ as a
 function of the asset price.}
\end{figure}

We now consider the sub-hedging strategies, i.e.~we look for $(N_t)$
and $g$ which would satisfy
\begin{equation*}
\dnt \leq N_T + g(S_T).
\end{equation*}
We design two such strategies. The first one is trivial as
it consists in doing nothing: we let $\lh^I\equiv 0$ and obviously
$\lh^I\leq \dnt$.  The second strategy involves holding cash, selling
a call and a put and entering a forward transaction upon the stock
price reaching a given level. More precisely, we have
\begin{equation}\label{eq:lH2}
\begin{split}
 \dnt \geq\ & 1 -\frac{(S_T-K_2)^+}{\ub-K_2}-\frac{(K_1-S_T)^+}{K_1-\lb}\\
 & + \frac{S_T-\ub}{\ub-K_2}\mathbf{1}_{H_{\ub}<H_{\lb}\wedge
   T} - \frac{S_T-\lb}{K_1-\lb}\mathbf{1}_{H_{\lb}<H_{\ub}\wedge
   T}=:\lh^{II}(K_1,K_2),
\end{split}
\end{equation}
where $\lb<K_1<K_2<\ub$ are arbitrary strikes. A graphical
representation of the inequality is given in Figure~\ref{fig:lH2}. The
inequality follows from the choice of the coefficients which are such
that
\begin{itemize}
\item on $\{T<H_{\lb}\wedge H_{\ub}\}$, i.e.\ on $\{\dnt=1\}$,
$\lh^{II}(K_1,K_2)\leq 1$ and $\lh^{II}(K_1,K_2)$ is equal to one on
$S_T\in [K_1,K_2]$, and is equal to zero for $S_T=\lb$ or $S_T=\ub$,
\item on $\{H_{\ub}<H_{\lb}\wedge T\}$, $\lh^{II}$ is non-positive and
is equal to zero on $\{S_T\geq K_2\}$ and
\item on $\{H_{\lb}<H_{\ub}\wedge T\}$, $\lh^{II}$ is non-positive and
is equal to zero on $\{S_T\leq K_1\}$.
\end{itemize}

\subsection{Model-free bounds on double no-touch options}
We have exhibited above several super- and sub- hedging strategies of the double no-touch option. They involved four stopping times and from now on \emph{we always assume} that they are included in $\stset$: 
\begin{equation}\label{eq:stset_assume}
\left\{0,H_{\lb},H_{\ub},\inf\{t<T\land H_{\ub}: S_t=\lb\}, \inf\{t<T\land H_{\lb}:S_t=\ub\}\right\}\subset\stset.
\end{equation}
It follows from the linearity of $\Pc$ that these induce bounds on the prices $\Pc(\dnt)$ admissible under no model-free arbitrage. More precisely, we have the following:
\begin{lemma}\label{lem:mf_bounds}
Let $\lb<S_0<\ub$ and suppose $\Pc$ admits no model-free arbitrage on $\mathcal{X}_D\cup\{\dnt\}$. Then 
\begin{equation}\label{eq:mf_bounds}
\begin{split}
\Pc(\dnt)&\leq \inf_{K_2,K_3\in \K,\ K_2> \lb, K_3<\ub}\left\{\Pc(\uh^{I}),\Pc(\uh^{II}(K_2)),\Pc(\uh^{III}(K_3))\right\},\\
\Pc(\dnt)&\geq \sup_{K_1,K_2\in\K,\ \lb<K_1<K_2<\ub}\left\{\Pc(\lh^{I}),\Pc(\lh^{II}(K_1,K_2))\right\}.
\end{split}
\end{equation}

\end{lemma}

\section{Model-free pricing of double no-touch options}\label{sec:prices}
In the previous section we exhibited a necessary condition
\eqref{eq:mf_bounds} for no model-free arbitrage and hence for
existence of a market model. Our aim in this section is to derive
sufficient conditions. In fact we will show that \eqref{eq:mf_bounds},
together with appropriate restrictions on call prices from
Propositions \ref{prop:wflvr} and \ref{prop:wa}, is essentially
equivalent to no WFLVR, or no WA when $\K$ is a finite set, and we can
then build a market model. Furthermore, we will compute explicitly the
supremum and infimum in \eqref{eq:mf_bounds}. To do this we have to
understand market models which are likely to achieve the bounds in
\eqref{eq:mf_bounds}. This is done using the technique of Skorokhod embeddings which we now discuss.

\subsection{The Skorokhod embedding problem}
Let $(B_t)$ be a standard real-valued Brownian motion with an
arbitrary starting point $B_0$. Let $\mu$ be a probability measure on
$\R$ with $\int_{\R}|x| \mu(dx)<\infty$ and $\int_{\R}x\mu(dx)=B_0$.

The Skorokhod embedding problem, $(SEP)$, is the following: given $(B_t), \mu$, find a stopping time $\tau$ such that the
stopped process $B_\tau$ has the distribution $\mu$, or simply:
$B_\tau \sim \mu$, and such that the process $(B_{t\wedge \tau})$ is
uniformly integrable (UI)\footnote{In some parts of the literature, the latter assumption is not included in the definition of the
problem, or an alternative property is used (see \cite{CoxHobson:06}). For the purposes of this article, we will
assume that all solutions have this UI property.}. We will often
refer to stopping times which satisfy this last condition as `UI
stopping times'. The existence of a solution was established by
Skorokhod \cite{Skorokhod:65}, and since then a number of further
solutions have been established, we refer the reader to \cite{Obloj:04b} for details.

Of particular interest here are the solutions by Perkins
\cite{Perkins:86} and the `tilted-Jacka' construction
(\cite{Jacka:88,Cox:04,Cox:05}). The Perkins embedding is defined in
terms of the functions
\begin{align}
\gp(x) & = \sup\left\{y < B_0 : \int_{(0,y)\cup (x,\infty)}
 (w-x) \, \mu(dw) \ge 0\right\}
\quad && x > B_0   \label{eq:gammaplus}\\
\gm(y) & = \inf \left\{x > B_0 : \int_{(0,y)\cup (x,\infty)} (w-y)
 \, \mu(dw) \le 0\right\}
&& y < B_0. \label{eq:gammaminus}
\end{align}
The Jacka embedding is defined in terms of the functions
$\Psi_\mu(x)$,$\Theta_\mu(x)$, where:
\begin{equation}
\label{eq:barycentres}
\Psi_\mu(K)=\frac{1}{\mu\big([K,\infty)\big)}\int_{[K,\infty)}x\mu(dx),\quad
\Theta_\mu(K)=\frac{1}{\mu\big((-\infty,K]\big)}\int_{(-\infty,K]}x\mu(dx),
\end{equation}
when (respectively) $\mu([K,\infty))$ and $\mu((-\infty,K])$ are
strictly positive, and $\infty$ and $-\infty$ respectively when these
sets have zero measure. The Perkins embedding is the stopping time
\begin{equation}\label{eq:perkins}
\tau_P:=\inf\big\{t: B_t\notin \big(\gp(\overline{B}_t),
\gm(\underline{B}_t) \big) \big\}.
\end{equation}
On the other hand, we define the `tilted-Jacka' stopping time as
follows. The `tilt' is to choose\footnote{In \cite{Jacka:88}, where
there is no `tilt', $K$ is chosen such that $(B_0 - \Theta_\mu(K)) =
(\Psi_\mu(K) - B_0)$.} $K \in (0,\infty)$, and set
\begin{eqnarray}
\tau_1 & := & \inf \left\{t \ge 0 : B_t \not\in( \Theta_\mu(K),
 \Psi_\mu(K))\right\}\label{eq:JackaTau1}\\
\tau_\Psi & = & \inf\left\{t \ge \tau_1 : \Psi_\mu(B_t) \le
 \sB_t\right\} \nonumber\\
\tau_\Theta & = & \inf\left\{t \ge \tau_1 : \Theta_\mu(B_t) \ge
 \iB_t \right\}. \nonumber
\end{eqnarray}
Then the `tilted-Jacka' stopping time is defined by:
\begin{equation}\label{eq:jacka}
\tau_J(K):= \tau_\Psi \indic{\tau_1 = \Psi_\mu(K)} +
\tau_\Theta\indic{\tau_1 = \Theta_\mu(K)}
\end{equation}
These embeddings are of particular interest due to their optimality
properties. Specifically, given any other stopping time $\tau$ which
is a solution to $(SEP)$, then we have the inequalities:
\begin{equation}\label{eq:perkins_opt}
\Pr\big(\underline{B}_\tau > \lb\big)\leq
\Pr\big(\underline{B}_{\tau_P} >\lb\big)\quad\textrm{and}\quad
\Pr\big(\overline{B}_\tau < \ub\big)\leq \Pr\big(\overline{B}_{\tau_P}<\ub\big).
\end{equation}
The embedding therefore `minimises the law of the maximum, and
maximises the law of the minimum.'

The `tilted-Jacka' embedding works the other way round. Fix $K \in
(0,\infty)$. Then if $\ub > \Psi_\mu(K)$ and $\lb < \Theta_\mu(K)$,
for all solutions $\tau$ to $(SEP)$, we have:
\begin{equation*}
\Pr\left( \lB_{\tau} > \lb\right) \ge \Pr\left( \lB_{\tau_J(K)} >
 \lb\right) \quad \mbox{ and } \quad
\Pr\left( \uB_{\tau} < \ub\right) \ge \Pr\left( \uB_{\tau_J(K)} <
 \ub\right).
\end{equation*}
\begin{remark} \label{rem:Kchoice} The importance of the role of $K$
can now be seen if we choose a function $f(x)$ such that $f(x)$ is
increasing for $x > B_0$ and decreasing for $x < B_0$. Since
$\Theta_{\mu}(K)$ and $\Psi_\mu(K)$ are both increasing, we can find
a value of $K$ such that (assuming suitable continuity)
$f(\Theta_\mu(K)) = f(\Psi_\mu(K))$. Then the `tilted-Jacka'
construction, with $K$ chosen such that $f(\Theta_\mu(K)) =
f(\Psi_\mu(K))$ maximises $\Pr(\sup_{t \le \tau}f(B_{t}) \ge z)$
over solutions of $(SEP)$ (see \cite{Cox:04} for details). We note
that the construction of \cite{Jacka:88} (where only a specific
choice of $K$ is considered) maximises $\Pr(\sup_{t \le \tau}|B_t|
\ge z)$. Of particular interest for our purposes will be the case
where $f(x) = \indic{x \not\in (\lb,\ub)}$, where $\lb < B_0 <
\ub$. In this case\footnote{At least if we assume the absence of
  atoms from the measure $\mu$. If the measure contains atoms, we
  need to be slightly more careful about some definitions, but the
  statement remains true. We refer the reader to \cite{CoxHobson:06}
  for details. (The proof of Theorem~14 therein is easily adapted to
  the centered case.)}, we can find $K$ either such that
$f(\Theta_\mu(K)) = f(\Psi_\mu(K))=0$, or such that
$f(\Theta_\mu(K)) = f(\Psi_\mu(K))=1$. In general this choice of $K$
will not be unique and we may take any suitable $K$. In fact, we may
classify which case we belong to: if\footnote{When there exists an
  interval to which $\mu$ assigns no mass, the inverse may not be
  uniquely defined. In which case, the argument remains true if we
  take $\Theta_{\mu}^{-1}(z) = \sup\{w \in \R : \Theta_\mu(w) \le
  z\}$ and $\Psi_\mu^{-1}(z) = \inf\{w \in \R : \Psi_\mu(w) \ge
  z\}$, i.e.\ we take $\Theta_{\mu}^{-1}$ left-continuous and
  $\Psi_{\mu}^{-1}$ right-continuous. In case $\mu$ has atoms at the
  end of the support, writing $\ovl\mu(x) = \mu([x,\infty))$:
  $1=\ovl\mu(a)>\ovl\mu(a+)$ or $\ovl\mu(b)>\ovl\mu(b+)=0$ this
  becomes slightly more complex as then we put
  $\Theta_\mu^{-1}(a)=\Theta_\mu^{-1}(a+)$ and
  $\Psi_\mu^{-1}(b)=\Psi_\mu^{-1}(b-)$ respectively.}
$\Theta_\mu^{-1}(\lb) \ge \Psi_{\mu}^{-1}(\ub)$ we can take $K \in
(\Psi_{\mu}^{-1}(\ub),\Theta_\mu^{-1}(\lb))$ (or equal to
$\Theta_\mu^{-1}(\lb)$ in the case where there is equality) which
has $f(\Theta_\mu(K)) = f(\Psi_\mu(K))=1$. Alternatively, if
$\Theta_\mu^{-1}(\lb) < \Psi_{\mu}^{-1}(\ub)$, then taking $K \in
(\Theta_\mu^{-1}(\lb), \Psi_{\mu}^{-1}(\ub))$ gives
$f(\Theta_\mu(K)) = f(\Psi_\mu(K))=0$. For such a choice of $K$, we
will call the resulting stopping time the tilted-Jacka embedding for
the barriers $\lb,\ub$.
\end{remark}

\begin{lemma}\label{lem:perkins}
For any $\lb<B_0<\ub$ and any stopping time $\tau$, which is a
solution to the Skorokhod embedding problem $(SEP)$ for $\mu$, we have
\begin{equation}\label{eq:perkins_double}
 \Pr\big(\underline{B}_\tau>\lb\textrm{ and } \overline{B}_\tau
 <\ub \big)\leq \Pr\big(\underline{B}_{\tau_P}>\lb\textrm{ and }
 \overline{B}_{\tau_P} <\ub \big),
\end{equation}
where $\tau_P$ is the Perkins solution \eqref{eq:perkins}.
\end{lemma}
\begin{proof}
We consider three possibilities:
\begin{enumerate}
\item First observe that we always have
 \begin{equation*}
   \Pr\big(\underline{B}_\tau>\lb,\ \overline{B}_\tau <\ub\big)\leq
   \Pr\big(B_\tau\in (\lb,\ub)\big)=\mu\big((\lb,\ub)\big).
 \end{equation*}
 From the definition \eqref{eq:perkins} of $\tau_P$ it follows that
 \begin{equation*}
   \Pr\big(\underline{B}_{\tau_P}>\lb,\ \overline{B}_{\tau_P}
   <\ub\big)=\mu\big((\lb,\ub)\big),
 \end{equation*}
 when $\ub\leq \gm(\lb)$ and $\gp(\ub)\leq \lb$ or when
 $\ub> \gm(\lb)$ and $\gp(\ub)> \lb$. The latter
 corresponds to $\mu((\lb,\ub))=1$ in which case a UI embedding
 always remains within $(\lb,\ub)$,
 i.e. $\Pr\big(\underline{B}_\tau>\lb,\ \overline{B}_\tau
 <\ub\big)= 1=\mu((\lb,\ub)).$
\item Suppose $\ub>\gm(\lb)$ and $\gp(\ub)\leq \lb$. We
 then have, using \eqref{eq:perkins_opt} and \eqref{eq:perkins},
 \begin{equation*}
   \Pr\big(\underline{B}_\tau>\lb,\ \overline{B}_\tau <
   \ub\big)\leq \Pr\big(\underline{B}_\tau > 
   \lb\big)\leq \Pr\big(\underline{B}_{\tau_P} >\lb\big) =
   \Pr\big(\underline{B}_{\tau_P}>\lb,\ \overline{B}_{\tau_P}
   <\ub\big).
\end{equation*}
\item Suppose $\ub\leq \gm(\lb)$ and $\gp(\ub)> \lb$. We
then have, using \eqref{eq:perkins_opt} and \eqref{eq:perkins},
\begin{equation*}
 \Pr\big(\underline{B}_\tau>\lb,\ \overline{B}_\tau < \ub\big) \leq
 \Pr\big(\overline{B}_\tau < \ub\big)\leq \Pr\big(
 \overline{B}_{\tau_P} < \ub\big) = \Pr\big(\underline{B}_{\tau_P}>
 \lb,\ \overline{B}_{\tau_P} <\ub\big).
 \end{equation*}
\end{enumerate}
\end{proof}
\begin{lemma}\label{lem:jacka}
For any $\lb<B_0<\ub$ and any stopping time $\tau$, which is a
solution to the Skorokhod embedding problem $(SEP)$ for $\mu$, we
have
\begin{equation}\label{eq:jacka_double}
 \Pr\big(\underline{B}_\tau>\lb\textrm{ and }\overline{B}_\tau <
 \ub \big) \geq \Pr\big(\underline{B}_{\tau_J(K)} >\lb \textrm{ and }
 \overline{B}_{\tau_J(K)} <\ub \big),
\end{equation}
where $\tau_J(K)$ is the tilted-Jacka embedding \eqref{eq:jacka} for
barriers $\lb,\ub$.
\end{lemma}

\begin{proof}
As noted in Remark~\ref{rem:Kchoice}, the tilted-Jacka embedding
with $f(x) = \indic{x \not\in (\lb,\ub)}$ corresponds to a choice of
$K$ such that $f(\Theta_\mu(K)) = f(\Psi_\mu(K))$. Suppose that both
these terms are one. Then $\Theta_\mu(K) < \lb$ and $\Psi_\mu(K) >
\ub$. In particular, since $B_{\tau_1} \in \{\Theta_\mu(K),
\Psi_\mu(K)\}$, we must have $\Pr\big(\underline{B}_{\tau_J(K)} >\lb
\textrm{ and } \overline{B}_{\tau_J(K)} <\ub \big) = 0$, and the
conclusion trivially follows.

Suppose instead that $f(\Theta_\mu(K)) = f(\Psi_\mu(K)) = 0$. Then
$\Theta_\mu(K) \ge \lb$ and $\Psi_\mu(K) \le \ub$. From the
definition of $\tau_{J}(K)$, paths will never cross $K$ after
$\tau_1$, so we have
\begin{eqnarray*}
 \Pr\big(\underline{B}_{\tau_J(K)} >\lb \textrm{ and }
 \overline{B}_{\tau_J(K)} <\ub \big) & = &
 1 - \Pr(\lB_{\tau_J(K)}\le \lb) - \Pr(\uB_{\tau_J(K)} \ge \ub)\\
 & = & 1 - \Pr(B_{\tau_J(K)} \le \Theta_{\mu}^{-1}(\lb)) -
 \Pr(B_{\tau_J(K)} \ge \Psi_{\mu}^{-1}(\ub)),
\end{eqnarray*}
where the second equality follows from the definitions of
$\tau_\Psi$ and $\tau_\Theta$. Finally, we note that the latter
expressions are exactly the maximal probabilities or the Az\'ema-Yor
and reverse Az\'ema-Yor embeddings (see \cite{AzemaYor:79} and
\cite{Obloj:04b}) so that for any solution $\tau$ to (SEP) for
$\mu$, we have
\begin{eqnarray*}
 \Pr(\uB_\tau \ge \ub) & \le & \Pr(B_{\tau_J(K)} \ge
 \Psi_{\mu}^{-1}(\ub))\\
 \Pr(\lB_\tau \le \lb) & \le & \Pr(B_{\tau_J(K)} \le
 \Theta_{\mu}^{-1}(\lb)).
\end{eqnarray*}
Putting these together, we conclude:
\begin{eqnarray*}
 \Pr\big(\underline{B}_\tau>\lb\textrm{ and }\overline{B}_\tau <
 \ub \big) & \ge & 1 - \Pr(\uB_\tau \ge \ub) - \Pr(\lB_\tau \le
 \lb)\\
 & \ge & 1 - \Pr(B_{\tau_J(K)} \le \Theta_{\mu}^{-1}(\lb)) -
 \Pr(B_{\tau_J(K)} \ge \Psi_{\mu}^{-1}(\ub))\\
 & \ge & \Pr\big(\underline{B}_{\tau_J(K)} >\lb \textrm{ and }
 \overline{B}_{\tau_J(K)} <\ub \big).
\end{eqnarray*}
\end{proof}

\subsection{Prices and hedges for the double no-touch option when
$\K=\R_+$}
We now have all the tools we need to compute the bounds in
\eqref{eq:mf_bounds} and prove they are the best possible bounds.  We
begin by considering the case where call options are traded at all
strikes: $\K=\R_+$.
\begin{theorem}\label{thm:wflvr}
Let $0<\lb<S_0<\ub$ and recall that $(S_t)$ has continuous paths in
$\pathsp$. Suppose $\mathcal{P}$ admits no WFLVR on
$\mathcal{X}_D$ defined via \eqref{eq:Xsetofassets}, \eqref{eq:XDdef} and \eqref{eq:stset_assume}. Then the following are equivalent: \\
(i) $\Pc$ admits no WFLVR on $\mathcal{X}_D\cup\{\dnt\}$,\\
(ii) there exists a $(\Pc,\mathcal{X}_D\cup\{\dnt\})$-market model,\\
(iii)\eqref{eq:mf_bounds} holds,\\
(iv) we have, with $\mu$ defined via \eqref{eq:mudef},
\begin{equation}\label{eq:upper_bound}
 \mathcal{P} \dnt\leq \min\left\{\iD(\lb) - \sD(\ub), \iD(\lb)
   +\frac{C(\gm(\lb))-P(\lb)}{\gm(\lb)-\lb}, 
   1-\sD(\ub)+\frac{P(\gp(\ub))-C(\ub)}{\ub-\gp(\ub)}\right\},
\end{equation}
where $\gamma_\mu^\pm$ are given in \eqref{eq:gammaplus}--\eqref{eq:gammaminus}, and
\begin{equation}\label{eq:lower_bound}
 \mathcal{P}\dnt \geq \left[1 - \frac{C(\Psi^{-1}_\mu(\ub))}{\ub -
   \Psi^{-1}_\mu(\ub)} -  \frac{P(\theta^{-1}_\mu(\lb))}{\theta^{-1}_\mu(\lb)-\lb}\right] \lor 0=\mu\Big(\big(\theta^{-1}_\mu(\lb) , \Psi^{-1}_\mu(\ub)
 \big)\Big)\lor 0,
\end{equation}
where $\theta_\mu,\Psi_\mu$ are given by \eqref{eq:barycentres}.

Furthermore, the upper bound in \eqref{eq:upper_bound} is attained
for the market model $S_t:=B_{\tau_P\wedge\frac{t}{T-t}}$, where
$(B_t)$ is a standard Brownian motion, $B_0=S_0$, and $\tau_P$ is
Perkins' stopping time \eqref{eq:perkins} embedding law $\mu$, where
$\mu$ is defined by \eqref{eq:mudef}.  The lower bound in
\eqref{eq:lower_bound} is attained for the market model
$S_t:=B_{\tau_J(K)\wedge\frac{t}{T-t}}$, where $(B_t)$ is a standard
Brownian motion, $B_0=S_0$, and $\tau_J(K)$ is the tilted-Jacka
stopping time \eqref{eq:jacka} embedding law $\mu$ for barriers
$\lb,\ub$.
\end{theorem}

\begin{remark} \label{rem:gammas}Note that the terms on the RHS of
\eqref{eq:upper_bound} correspond respectively to $\Pc \uh^{I}$,
$\Pc\uh^{II}(\gm(\lb))$ and $\Pc\uh^{III}(\gp(\ub))$. From the
proof, we will be able to say precisely which term is the smallest:
the second term is the smallest if $\ub>\gm(\lb)$ and $\gp(\ub)\leq
\lb$, the third term is the smallest if $\ub\leq \gm(\lb)$ and
$\gp(\ub)> \lb$ and otherwise the first term is the smallest. \\
The lower bound \eqref{eq:lower_bound} is non-zero if and only if $\theta^{-1}_\mu(\lb)<\Psi^{-1}_\mu(\ub)$ and is then equal to $\Pc\lh^{II}(\theta^{-1}_\mu(\lb),\Psi^{-1}_\mu(\ub))$.\\
It follows from the definitions of $\gamma_\mu^\pm$ (see \cite{Perkins:86}) that we have
\begin{equation}\label{eq:defgammas_2}
\begin{split}
\frac{C(\gm(\lb))-P(\lb)}{\gm(\lb)-\lb}&=\inf_{K>S_0}\frac{C(K)-P(\lb)}{K-\lb}\\
\frac{P(\gp(\ub))-C(\ub)}{\ub-\gp(\ub)}&=\inf_{K<S_0}\frac{P(K)-C(\ub)}{\ub-K}.
\end{split}
\end{equation}
Using this and the remarks on when respective terms in \eqref{eq:upper_bound} are the smallest, we can rewrite the minimum on the RHS of \eqref{eq:upper_bound} as
\begin{equation}\label{eq:upper_bound_2}
\min\left\{\iD(\lb) - \sD(\ub), \iD(\lb)
   +\inf_{K\in (S_0,\ub)}\frac{C(K)-P(\lb)}{K-\lb}, 
   1-\sD(\ub)+\inf_{K\in (\lb,S_0)}\frac{P(K)-C(\ub)}{\ub-K}\right\}.
\end{equation}

\end{remark}
\begin{remark}\label{rem:cont_paths} We want to investigate briefly
what happens if the assumptions on continuity of paths are relaxed,
namely if the bounds \eqref{eq:upper_bound} and
\eqref{eq:lower_bound}, or equivalently \eqref{eq:mf_bounds}, are
still consequences of no model-free arbitrage. If we consider the
upper barrier \eqref{eq:upper_bound}, the assumption that $S_t$ is
continuous may be relaxed slightly: we only require that the price
does not jump across either barrier, but otherwise jumps may be
introduced.  We can also consider the general problem without any
assumption of continuity, however this becomes fairly simple: the
upper bound is now simply $\mu((\lb,\ub))$, which corresponds to the
price process $(S_t)_{0 \le t \le T}$ which is constant for $t \in
[0,T)$, so $S_t = S_0$, but has $S_T \sim \mu$. \\
Unlike the upper bound, the continuity or otherwise of the process
will make little\footnote{One has to pay some attention to avoid
 measure-theoretic problems as some (minimal) assumption on the process and the filtration are required to guarantee that the first entry time into $[0,\lb]$ is a stopping time.} 
 difference to the lower bound. Clearly, the lower bound is still attained by the same
construction, however, we can also show that a similar subhedge to
\eqref{eq:lH2} still holds. In fact, the only alteration that is
needed in the discontinuous case concerns the forward purchase. We
construct the same initial hedge, but suppose now the asset jumps
across the upper barrier, to a level $z$ say. Then we may still buy
$1/(\ub-K_2)$ forward units, but the value of this forward at
maturity is now: $(S_T -z)/(\ub-K_2)$. So the difference between
this portfolio at maturity and the payoff given in \eqref{eq:lH2} is
just
\begin{equation*}
 \frac{S_T-z}{\ub-K_2} - \frac{S_T-\ub}{\ub-K_2} = \frac{\ub-z}{\ub-K_2}
\end{equation*}
which is negative, and therefore the strategy is still a
subhedge. Consequently, the same lower bound is valid.
\end{remark}

\begin{proof}
If $\Pc$ admits no WFLVR on $\mathcal{X}_D\cup\{\dnt\}$ then it also
admits no model-free arbitrage and $(i)$ implies $(iii)$ with Lemma
\ref{lem:mf_bounds}.  As observed above, the three terms on the RHS
of \eqref{eq:upper_bound} are respectively $\mathcal{P} \uh^I$,
$\mathcal{P}
\uh^{II}(\gm(\lb))$ and $\mathcal{P}\uh^{III}(\gp(\ub))$, and the
two terms on the RHS of \eqref{eq:lower_bound} are
$\Pc\lh^{II}(\theta^{-1}_\mu(\lb),\Psi^{-1}_\mu(\ub))$ and
$\Pc\lh^I=0$. Thus clearly $(iii)$ implies $(iv)$. \\
Let $\mu$ be defined via \eqref{eq:mudef}, $(B_t)$ a Brownian motion
defined on a filtered probability space and $\tau_P$, $\tau_J$
respectively Perkins' \eqref{eq:perkins} and (tilted) Jacka's
\eqref{eq:jacka} stopping times embedding $\mu$ for barriers
$\lb,\ub$. Since $\Pc$ admits no WFLVR on $\mathcal{X}_D$,
Proposition \ref{prop:wflvr} implies that both
$S^P_t:=B_{\tau_P\wedge\frac{t}{T-t}}$ and
$S^J_t:=B_{\tau_J\wedge\frac{t}{T-t}}$ are market models matching
$\Pc$ on $\mathcal{X}_D$.  It follows from the proof of Lemma
\ref{lem:perkins} that $\E\mathbf{1}_{\overline{S}^P_T< \ub,\,
 \underline{S}^P_T> \lb}$ is equal to the RHS of
\eqref{eq:upper_bound}. Likewise, it follows from the proof of Lemma
\ref{lem:jacka} that $\E\mathbf{1}_{\overline{S}^J_T< \ub,\,
 \underline{S}^J_T> \lb}$ is equal to the RHS of
\eqref{eq:lower_bound}.  Enlarge the filtration of $(B_t)$ initially
with an independent random variable $U$, uniform on $[0,1]$, let
$\tau_\lambda=\tau_P\indic{U\leq \lambda}+\tau_J\indic{U>\lambda}$
and $S^\lambda_t:=B_{\tau_\lambda\wedge\frac{t}{T-t}}$. Then
$(S^\lambda_t)$ is a $(\Pc,\mathcal{X}_D)$-market model and
$\E\mathbf{1}_{\overline{S}^\lambda_T< \ub,\,
 \underline{S}^\lambda_T> \lb}$ takes all values between the bounds
in \eqref{eq:upper_bound}-\eqref{eq:lower_bound} as $\lambda$ varies
between $0$ and $1$. We conclude that $(iv)$ implies
$(ii)$. Obviously we have $(ii)$ implies $(i)$.
\end{proof}

\subsection{Pricing and hedging when $\K$ is finite}
In practice, the assumption that call prices are known for every
strike is unrealistic, so we consider now the case when $\K$ is
finite. The only assumption we make, which is satisfied in most market
conditions, is that there are enough strikes to separate the
barriers. Specifically, we shall make the following assumption:
\begin{description}
\item[{\bf(A)}] $\K=\{K_0,K_1,\ldots,K_N\}$, $0=K_0<K_1<\ldots<K_N$,
and the barriers $\lb$ and $\ub$ satisfy: $\lb > K_2$, there are at
least 2 traded strikes between $\lb$ and $\ub$, and $\ub < K_{N-1}$,
with $S_0>C(K_1)>C(K_{N-1})>0$.
\end{description}
It is more convenient to consider the lower and the upper bounds on
the price of the double no-touch option independently. The upper bound involves
digital calls and when these are not traded in the market the results
are somewhat technical to formulate. We start with the lower bound
which is relatively straightforward.
\begin{theorem} \label{thm:finstrikelower} 
Recall $\mathcal{X}$ defined via \eqref{eq:Xsetofassets} and \eqref{eq:stset_assume}.
Suppose \textbf{(A)} holds and $\mathcal{P}$ admits no WA on $\mathcal{X}\cup\{\dnt\}$. Then
\begin{equation}\label{eq:lower_bound_fs}
 \mathcal{P} \dnt \ge \max_{i \le j: \lb \le K_i \le K_j \le \ub}
 \left[ 1 - \frac{C(K_j)}{\ub - K_j} -
   \frac{P(K_i)}{K_i - \lb}\right]\lor 0\ .
\end{equation}
The bound is tight --- there exists a
$(\mathcal{P},\mathcal{X})$-market model, under which the above
price is attained.
\end{theorem}

\begin{proof}
The bound is just a rewriting of the lower bound in
\eqref{eq:mf_bounds} in which we omitted $\Pc \lh^I=0$.  It remains
to construct a market model under which the bound is attained.

Recall that $C(K_0)=C(0)=S_0$ and choose $K_{N+1}$ such that
\begin{equation*}
 K_{N+1} \ge \max\left\{K_N + C(K_N)\frac{K_N - K_{N-1}}{C(K_{N-1}) -
     C(K_N)}, \ub+1 \right\},
\end{equation*}
and set $C(K_{N+1}) = 0$. This extension of call prices preserves
the no WA property and by Proposition \ref{prop:calls_WA} we may
extend $C$ to a function on $\R_+$ satisfying
\eqref{eq:Cbasic}--\eqref{eq:Ccnv}. In fact, we may take $C$ to be
linear in the intervals $(K_i,K_{i+1})$ for $i\le N$, and setting
$C(K) = 0$ for $K\ge K_{N+1}$. To $C$ we associate a measure $\mu$
by \eqref{eq:mudef} which has the representation
\begin{equation*}
 \mu = \sum_{i=0}^{N+1} (C_{+}'(K_{i}) - C_{-}'(K_{i})) \delta_{K_i}
\end{equation*}
where we take $C_{-}'(0) = 1$. Note that by \textbf{(A)} the
barriers $\lb,\ub$ are not at the end of the support of $\mu$. For
the definition of $\Theta_\mu$ and $\Psi_\mu$ it follows that their
inverses (which we took left- and right- continuous respectively)
take values in $\K$. The theorem now follows from Theorem
\ref{thm:wflvr} using the equivalence of $(iii)$ and $(iv)$. Note
that \eqref{eq:lower_bound} gives precisely the traded strikes
$K_i,K_j$ for which the maximum in the RHS of
\eqref{eq:lower_bound_fs} is attained.
\end{proof}

We now consider the upper bound. There are several issues that will
make this case more complicated than the previous. Wheras in the lower
bound, we are purchasing only call/put options at strikes between
$\lb$ and $\ub$, in the upper bound, we need to consider how to infer
the price of a digital option at $\lb$ or $\ub$, and consider the
possibility that there are no options traded exactly at the strikes
$\lb$ and $\ub$. Secondly, the upper bound will prove to be much more
sensitive to the discontinuity in the payoff of the
double no-touch. This is because, when there are only finitely many
strikes, the measure $\mu$ -- the market model law of $S_T$ -- is not
specified and in order to maximise $\E \dnt$, one wants to have as
many paths as possible finishing as near to $\ub$ and $\lb$ within the
constraints imposed by the calls; to do this, we want to put atoms of
mass `just to the right of $\lb$', and `just to the left of
$\ub$'. For this reason, in the final case we consider, for some
specifications of the prices, the upper bound cannot always be
attained under a suitable model, but rather, in general can only be
arbitrarily closely approximated. These issues would not arise if we
were to consider modified double no-touch option with payoff
$\mathbf{1}_{\overline{S}_T\leq \ub,\, \underline{S}_T\geq \lb}$.

We begin by considering the simpler case where there are calls and
digital calls traded with strikes $\ub$ and $\lb$:
\begin{theorem} \label{thm:ubfinstrike1} Recall $\mathcal{X}_D$
defined via \eqref{eq:Xsetofassets}, \eqref{eq:XDdef} and
\eqref{eq:stset_assume}. Suppose \textbf{(A)} holds,
$\lb,\ub\in\K$ and $\Pc$ admits no weak arbitrage on
$\mathcal{X}_D$ and no model-free arbitrage on
$\mathcal{X}_D\cup\{\dnt\}$.  Then the price of the double
no-touch option is less than or equal to
\begin{equation}\label{eq:ub2}
  \min\left\{\iD(\lb) - \sD(\ub),
    \iD(\lb) + \inf_{K\in\K\cap (S_0,\ub]} \frac{C(K)-P(\lb)}{K-\lb},
    (1-\sD(\ub)) + \inf_{K\in\K\cap [\lb,S_0)} \frac{P(K)-C(\ub)}{\ub-K}\right\}.
\end{equation}
Further, there exists a sequence of market-models for
$(\Pc,\mathcal{X}_D)$ which approximate the upper bound, and
if $\mathcal{P}$ attains the upper bound, then either there exists a
weak arbitrage on $\mathcal{X}_D\cup\{\dnt\}$, or there exists a
$(\mathcal{P},\mathcal{X}_D\cup\{\dnt\})$-market model.
\end{theorem}

\begin{proof}
That \eqref{eq:ub2} is an upper bound is a direct consequence of
Lemma \ref{lem:mf_bounds} and the three terms correspond to the
three terms on the RHS of the upper bound in
\eqref{eq:mf_bounds}. Since there is no WA, by Lemma~\ref{lem:wkarb},
we can extend $C$ to a piecewise linear function on $\R_+$ which
satisfies \eqref{eq:Cbasic}--\eqref{eq:Ccnv} and
\eqref{eq:D_basic}. More precisely, let $i,j$ be such that
$K_i<\lb<K_{i+1}$ and $K_j<\ub<K_{j+1}$. Then we can take $C$
piecewise linear with kinks for $K\in \K\cup\{\lK,\uK\}$, where
$K_i<\lK<\lb$ and $K_j<\uK<\ub$ can be chosen arbitrary close to the
barriers. Note also that, using \eqref{eq:D_basic}, we have then
$C(\uK)=C(\ub)+\sD(\ub)(\ub-\uK)$, with a similar expression for
$C(\lK)$.  Consider the associated market model of Theorem
\ref{thm:wflvr} which achieves the upper bound
\eqref{eq:upper_bound}. We now argue that we can choose $\lK,\uK$ so
that \eqref{eq:ub2} approximates \eqref{eq:upper_bound} arbitrarily
closely. Since $C$ is piecewise linear we have that
\eqref{eq:upper_bound}, which can be expressed as
\eqref{eq:upper_bound_2}, is equal to \eqref{eq:ub2} but with $\K$
replaced by $\K\cup\{\lK,\uK\}$ and we just have to investigate
whether the addition of two strikes changes anything. We investigate
$\uK$, the case of $\lK$ is similar.  Note that we can make
$f(\uK):=(C(\uK)-P(\lb))/(\uK-\lb)$ as close to $f(\ub)$ as we want
by choosing $\uK$ sufficiently close to $\ub$. Hence if the minimum
in \eqref{eq:ub2} is strictly smaller than $\iD(\lb)+f(\ub)$ then
the addition of $\uK$ does not affect the minimum, and we note
further that in such a case, we may construct a market model
(assuming similar behaviour at $\lK$) using values of $\lK, \uK$
sufficiently close to $\lb, \ub$ respectively. Otherwise, suppose
the minimum in \eqref{eq:ub2} is achieved by
$\iD(\lb)-\sD(\ub)=\iD(\lb)+f(\ub)$. Then we have
$$
f(\uK)=\frac{C(\ub)+\sD(\ub)(\ub-\uK)-P(\lb)}{\uK-\lb}=
\frac{C(\ub)-P(\lb)}{\uK-\lb}-f(\ub) \frac{\ub-\uK}{\uK-\lb}=f(\ub),
$$
and hence the minimum in \eqref{eq:upper_bound_2} is also attained
by the first term and is equal to \eqref{eq:ub2}. Again, the
extension of $C$ allows us to construct a suitable market
model. Finally, consider the case where $f(\ub) < -\sD(\ub)$, and
the second term at $\ub$ is indeed the value of \eqref{eq:ub2}. Then
taking a sequence of models as described above, with $\lK, \uK$
converging to $\lb, \ub$ respectively, we get a suitable
approximating sequence. We finally show that in this case, if $\Pc$ prices double no-touch at \eqref{eq:ub2} then there is a weak arbitrage: suppose $\Pr$ is a model with $\Pr(H_{\lb} <
T, S_T \in (\lb,\ub)) > 0$. Then we can purchase the superhedge
$\ol{H}^{II}(\ub)$ and sell the double no-touch for zero initial net
cash flow, but with a positive probability of a positive reward (and
no chance of a loss). For all other models $\Pr$, we purchase a portfolio which is short $\frac{1}{\ub-\lb}$ puts
at $\lb$, long $\frac{1}{\ub-\lb}$ calls at $\ub$, and long the digital call at $\ub$; if the process hits $\lb$, we sell forward
$\frac{1}{\ub-\lb}$ units of the underlying. This has negative setup cost, since $f(\ub) < -\sD(\ub)$, and zero probability of a loss as now $\Pr(H_{\lb} < T, S_T \in
(\lb,\ub)) = 0$.
\end{proof}
\psfrag{b1}{$\lb$}
\psfrag{b2}{$\lb'$} \psfrag{ki1}{$K_{i-1}$} \psfrag{ki2}{$K_{i}$}
\psfrag{ki3}{$K_{i+1}$} \psfrag{ki4}{$K_{i+2}$}
The general case, where we do not assume that calls trade at the barriers, nor that there are suitable digital options, is slightly
more complex. The key point to understanding this case is to consider
models (or extensions of $C(\cdot)$ to the whole of $\R_+$) which
might maximise each of the individual terms (at $\ub$ or $\lb$) in
\eqref{eq:ub2}, and which agree with the call prices. Observe for example that at $\lb$ it is optimal to minimise the call price
$C(\lb)$, and also maximise $\iD(\lb)$ (at least for the first two terms in \eqref{eq:ub2}). If we convert this to a statement
about the call prices $C(\cdot)$ the aim becomes: minimise $C(\lb)$, and maximise $-C_{+}'(\lb)$. It is easy to see that choosing the
smallest value of $C(\lb)$ which maintains the convexity we also maximise $-C_-'(\lb)$. However this does not
quite work for $-C_{+}'(\lb)$, although we will `almost' be able to use it. It turns out (see
Lemma~\ref{lem:wkarb}) that even the non-attainable lower bound, corresponding to taking $-C_{-}'(\lb) = -C_{+}'(\lb)$, is still
consistent with no model-free arbitrage, but it is not consistent with no weak arbitrage. However, we will be able to find a sequence of
models under which the prices do converge to the optimal set of values. 

We now consider the bounds in more detail. Suppose $i$ is such that
$K_i \le \lb \le K_{i+1}$. By convexity, the value of $C(\lb)$ must lie
above the line passing through $\{(K_{i-1}, C(K_{i-1})),(K_{i},
C(K_{i}))\}$, and also the line passing through $\{(K_{i+1},
C(K_{i+1})),(K_{i+2}, C(K_{i+2}))\}$. So to minimise $C(\lb)$ we let it be:
\begin{equation}\label{eq:Clbdefn}
C(\lb) = \max\left\{C(K_i) + \frac{C(K_i) - C(K_{i-1})}{K_i -
   K_{i-1}}(\lb - K_i),
 C(K_{i+1}) + \frac{C(K_{i+1}) - C(K_{i+2})}{K_{i+2} -
   K_{i+1}}(K_{i+1}-\lb) \right\}
\end{equation}
and we set the corresponding `optimal' digital call price:
\begin{equation} \label{eq:dlbdefn}
\iD(\lb) = -\frac{C(\lb) - C(K_i)}{\lb-K_i}.
\end{equation}

The prices for the third term in \eqref{eq:ub2} are derived in a
similar manner: if we suppose $j \ge i+2$ (by assumption {\bf (A)})
is such that $K_j \le \ub \le K_{j+1}$, the resulting prices at $\ub$
are:
\begin{equation}\label{eq:Cubdefn}
C(\ub) = \max\left\{C(K_j) + \frac{C(K_j) - C(K_{j-1})}{K_j -
   K_{j-1}}(\ub - K_j),
 C(K_{j+1}) + \frac{C(K_{j+1}) - C(K_{j+2})}{K_{j+2} -
   K_{j+1}}(K_{j+1}-\ub) \right\}
\end{equation}
and
\begin{equation} \label{eq:dubdefn}
\sD(\ub) = -\frac{C(K_{j+1}) - C(\ub)}{K_{j+1}-\ub}.
\end{equation}
We note that assumption {\bf (A)} is necessary here to ensure that
the extended prices are free of arbitrage. Otherwise, we would not in general be able to add assets to the initial market in a
way that is consistent with \eqref{eq:Cbasic}.

\begin{theorem} \label{thm:ubfinstrike} Recall $\mathcal{X}$,
$\mathcal{X}_D$ defined via \eqref{eq:Xsetofassets},
\eqref{eq:XDdef} and \eqref{eq:stset_assume}.  Suppose \textbf{(A)}
holds, $\lb,\ub\notin \K$, and $\Pc$ admits no weak arbitrage on
$\mathcal{X}$.  Define the values of $C(\cdot)$ and $D(\cdot)$ at
$\lb,\ub$ respectively via
\eqref{eq:Clbdefn}--\eqref{eq:dubdefn}. Then if $\Pc$ admits no
model-free arbitrage on $\mathcal{X}_D\cup\{\dnt\}$, the price of the double
no-touch option is less than or equal to
\begin{equation}\label{eq:ub2B}
\left\{\iD(\lb) - \sD(\ub), \iD(\lb) + \inf_{K \in \mathbb{K} \cap (S_0,\infty)} \frac{C(K)-P(\lb)}{K-\lb}, (1-\sD(\ub)) + \inf_{K \in \mathbb{K} \cap [0,S_0)}
   \frac{P(K)-C(\ub)}{\ub-K}\right\}.
\end{equation}
Further, there exists a sequence of $(\Pc,\mathcal{X})$-market models which approximate the upper bound.
Finally, if $\Pc$ attains the upper bound, and when extended via \eqref{eq:Clbdefn}--\eqref{eq:dubdefn} admits no WA
on $\mathcal{X}_D\cup\{\dnt\}$ for $\K\cup\{\lb,\ub\}$ then there
exists a $(\mathcal{P},\mathcal{X}_D\cup\{\dnt\})$-market model.
\end{theorem}
We defer the proof to Appendix \ref{ap:weakarb}. Note that exact conditions determining whether the no WA property is
met are given in Lemma~\ref{lem:wkarb}. As will be clear from the proof, we could take $K\in \K\cap (S_0,K_{j+1}]$ and $K\in \K\cap [K_i,S_0)$ in the second and third terms in \eqref{eq:ub2B} respectively.  However, unlike in previous theorems, we need to include strikes $K_i$ and $K_{j+1}$ as we don't have the barriers as traded strikes.

This result is our final theorem concerning the structure of the
option prices in this setting. We want to stress the fairly pleasing
structure that all our results exhibit: we are able to exactly specify
prices at which the options may trade without exhibiting model-free
arbitrage. Moreover, we are able to specify the cases where there
exist market models for a given set of prices. In general, the two
sets are exclusive, and with the possible exception of a boundary
case, constitute all prices. On the boundary, if there is no model, we
are able to show the existence of an arbitrage of a weaker form than
the model-free arbitrage.

\section{Applications}\label{sec:app}
We turn now to possible applications of our results and present some
numerical simulations. We keep the discussion here rather brief and
refer the reader to our paper on double touch options \cite{CoxObloj:08} for more
details on implementation and application of robust hedging arguments.

The first natural application is for pricing. Namely, seeing call
prices in the market, we can instantly deduce robust price bounds on
the double no-touch options using Theorems \ref{thm:finstrikelower} and
\ref{thm:ubfinstrike1}. However, typically these bounds are too wide
to be of any practical use. This will be the case for example in
foreign exchange markets, where double no-touch options are liquid and bid-aks
spreads are very small. In fact in major currency pairs, these options
are so liquid that the price is given by the market --- i.e.\ should be
treated as an input to the model, see Carr and Crosby
\cite{CarrCrosby:08}.

The second application is for robust hedging --- and this is where we
believe our techniques can be competitive. Standard delta/vega hedging
techniques for double no-touch options face several difficulties, such
as:
\begin{itemize}
\item \emph{model risk} -- model mis-specification can result in incorrect
hedges,
\item \emph{transaction costs} -- these can run high as vega hedging is
expensive,
\item \emph{discrete monitoring} -- in practice hedges can only be updated
discretely and the more often they are updated the larger the
transaction costs,
\item \emph{gamma exposure} -- when the option is close to the barrier close
to maturity the delta is growing rapidly, in practice the trader
then stops delta-hedging and takes a view on the market.
\end{itemize}
Our robust hedges provide a simple alternative which avoids all of the
above-listed problems. Specifically, say a trader sells a double
no-touch option struck at $(\lb,\ub)$ for a fair premium $p$. She can
then set up one of our super-hedges $\uh^{i}(K)$, for $i=I,II,III$ for
a premium $\Pc\uh^{i}(K)$ which will be typically larger then $p$. The
superhedge then requires just that she monitors if the barriers are
crossed and if so that she buys or sells appropriate amounts of
forwards. Then at maturity $T$ her portfolio (hedging error) is
$$X=\uh^{i}(K)-\dnt\geq 0-\Pc\uh^{i}(K)+p,$$
which has zero expectation and is bounded below by
$p-\Pc\uh^{i}(K)$. Depending on the risk aversion and gravity of
problems related to delta/vega-hedging listed above, this may be an
appealing way of hedging the double no-touch option.\\
We give a simple example. Consider the following Heston model (based
on the parameter estimates for USD/JPY given in \cite{CarrWu:07}):
\begin{equation}
\label{eq:heston}
\left\{
\begin{array}{rclcl}
  dS_t&=&\sqrt{v_t}S_tdW^1_t, & & S_0=S_0,\ v_0=\sigma_0\\
  dv_t&=&\kappa(\theta-v_t)dt+\xi\sqrt{v_t}dW^2_t, &&
  d\langle W^1,W^2\rangle_t=\rho dt,
\end{array}
\right.
\end{equation}
with parameters 
\begin{equation}
\label{eq:heston_par}
S_0=2.006,\ \sigma_0=0.025,\ \kappa=0.559,\ \theta=0.02,\ \xi=0.26\textrm{ and }\rho=0.076,
\end{equation}
and a double no-touch option with 6 month maturity struck at
$\lb=1.95$, $\ub=2.05$. The numerically evaluated fair price of this
option in this model is $p=0.3496$. The cheapest of our superhedges is
$\uh^{I}$, which was just a digital option paying $1$ when
$1.9364<S_T< 2.0636$ (the closest strikes available to the
barriers). The most expensive subhedge was $\lh^{I}$ which consisted
in doing nothing. We compare outcomes of two hedging approaches of a
short position in a double no-touch option: standard delta/vega
hedging but using BS deltas (with at-the-money implied volatility) and
the robust approach outlined above. We assume proportional transaction
costs of $1\%$ when trading in calls or puts and $0.02\%$ when trading
in the underlying. The delta/vega hedge is rebalanced daily and we
stop hedging when the deltas are too large, more precisely we stop
hedging when the transaction costs associated with re-balancing the
hedge are above $0.02$. The distribution of hedging errors from both
strategies over 10000 Monte Carlo runs is given in Figure
\ref{fig:hedgeerr}.
\begin{figure}[htbp]
\label{fig:hedgeerr}
\begin{center}
\includegraphics[height=6.5cm]{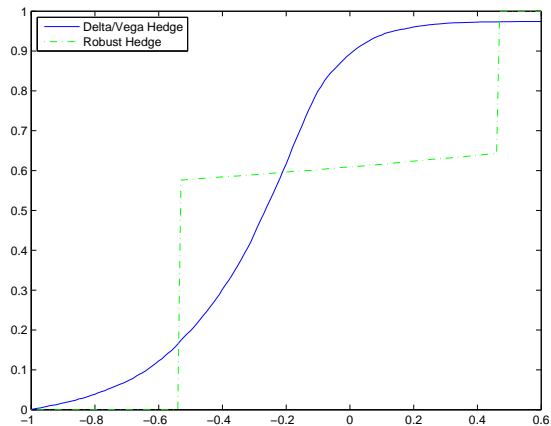}
\end{center}
\caption{Cumulative distributions of hedging errors of a short position in a double no-touch option using delta/vega hedging and robust hedging.}
\end{figure}
The advantage of a $X$ being bounded below is clearly visible. The
average transaction costs were $0.32$ and $0.13$ for delta/vega and
robust hedges respectively which shifted the average hedging errors
for the delta/vega hedging to the left. In consequence, an exponential
utility trader would prefer our robust hedge with utility $-0.2755$
against $-0.4327$ from delta/vega hedging errors.

Naturally, there are a number of ways of improving on the standard
delta/vega hedge, as well as other robust hedging approaches --- see
the discussion in \cite{CoxObloj:08}.  More involved numerical
analysis would be necessary to judge whether our strategies can still
outperform, and in what sense, these improved hedging strategies. In
light of the above results, this seems like an interesting practical
issue to pursue. Finally, we note that hedging a single barrier option
in any way is impractical due to transaction and operational costs and
in practise banks hedge huge portfolios of barrier options rather than
a single option. An adaptation of our techniques would be thus
required before they could be implemented in real markets.

\appendix
\section{Weak Arbitrage}
\label{ap:weakarb}

In this appendix we are interested in the possible extensions of the
pricing operator from an arbitrage-free set of call prices $C(K_i)$,
$0 = K_0 < K_1 < \ldots < K_N$ to include a call $C(\lb)$ and a
digital call option $\iD(\lb)$ where $\lb \in (K_{i},K_{i+1})$, and a call $C(\ub)$ and a digital call $\sD(\ub)$, where $\ub
\in (K_{j},K_{j+1})$ for {$j \ge i+2$.}

To avoid simple arbitrages if we add in a call option with strike at $x \in
(K_{l},K_{l+1})$, we need $C(x)$ to satisfy:
\begin{eqnarray}
C(x) & \le & C(K_{l}) + (x-K_{l})\frac{C(K_{l+1}) -
 C(K_{l})}{K_{l+1}-K_{l}} \label{eq:Cup}\\
C(x) & \ge & C(K_{l}) +
(x-K_{l})\frac{C(K_{l})-C(K_{l-1})}{K_{l}-K_{l-1}}\label{eq:Clow1}\\
C(x) & \ge & C(K_{l+1}) -
(K_{l+1}-x)\frac{C(K_{l+2})-C(K_{l+1})}{K_{l+2}-K_{l+1}} \label{eq:Clow2}
\end{eqnarray}
so we require these bounds to hold with $(x,l) = (\lb,i)$ and $(x,l) =
(\ub,j)$. Some further simple arbitrages imply $\iD(x)$ and $\sD(x)$ satisfy
\begin{equation}\label{eq:Dbounds}
\frac{C(x) - C(K_{l+1})}{K_{l+1}-x} \le \iD(x) \le \sD(x) \le \frac{C(K_{l})-
 C(x)}{x-K_{l}}.
\end{equation}
Depending on $x$, one of the lower bounds \eqref{eq:Clow1} or
\eqref{eq:Clow2} may be redundant. Specifically, there exists $b_* \in
(K_{i},K_{i+1})$ such that the right-hand sides agree for $x = b_*$,
\eqref{eq:Clow1} is larger for $x < b_*$, and \eqref{eq:Clow2} is
larger for $x > b_*$, and there is a similar point $b^*\in(K_{j},K_{j+1})$
at which the corresponding versions of \eqref{eq:Clow1} and
\eqref{eq:Clow2} agree.

We now prove a result which connects the traded prices of these
options, the existence of both model-free arbitrages and weak
arbitrages, and the existence of a model which agrees with a given
pricing operator:

\begin{lemma} \label{lem:wkarb} Recall $\mathcal{X}$ given via
\eqref{eq:Xsetofassets} and \eqref{eq:stset_assume}.  {Suppose
  $\mathbb{K}$ is finite $\lb, \ub \not\in \mathbb{K}$, {\bf (A)}
  holds, and $\mathcal{P}$ admits no WA on $\mathcal{X}$. Let  $\mathcal{X}_D = \mathcal{X} \cup \{(S_T-\lb)^+,\indic{S_T > \lb},
  (S_T-\ub)^+, \indic{S_T \ge \ub}\}$.}  Then if $C(\lb), \iD(\lb),
C(\ub), \sD(\ub)$ satisfy \eqref{eq:Cup}--\eqref{eq:Dbounds}, there
exists an extension of $\mathcal{P}$ to $\mathcal{X}_D$ with
\[
\mathcal{P} \indic{S_T > \lb} = \iD(\lb), \quad \mathcal{P} (S_T -
\lb)^+ = C(\lb), \quad \mathcal{P} \indic{S_T \ge \ub} = \sD(\ub), \quad
\mathcal{P} (S_T - \ub)^+ = C(\ub)
\]
such that $\mathcal{P}$ admits no model-free arbitrage. Conversely,
if any of \eqref{eq:Cup}--\eqref{eq:Dbounds} fail, then there exists
a model-free arbitrage.

Moreover, if there is no model-free arbitrage, there is a
weak-arbitrage if and only if either $\lb \ge b_*$,
\begin{equation} \label{eq:wa1}
 C(\lb) = C(K_{i+1}) - (K_{i+1} -\lb)\frac{C(K_{i+2}) -C(K_{i+1})}{K_{i+2}
   -K_{i+1}} \mbox{ and } \frac{C(\lb) - C(K_{i+1})}{K_{i+1}-\lb} <
 \iD(\lb) \le \frac{C(K_{i})- C(\lb)}{\lb-K_{i}},
\end{equation}
or $\ub \le b^*$,
\begin{equation} \label{eq:wa2}
 C(\ub) =  C(K_{j}) +
 (\ub-K_{j})\frac{C(K_{j})-C(K_{j-1})}{K_{j}-K_{j-1}} \mbox{
   and } \frac{C(\ub) - C(K_{j+1})}{K_{j+1}-\ub} \le \sD(\ub) < \frac{C(K_{j})-
   C(\ub)}{\ub-K_{j}}.
\end{equation}
Finally, if there is no weak arbitrage, then there exists a
$(\mathcal{P}, \mathcal{X_D})$-market model. Furthermore, we can take
the model such that $C(\cdot)$ is piecewise linear and has kinks only
at
$K_1<\ldots<K_i<\lb<\lK<K_{i+1}<\ldots<K_j<\uK<\ub<K_{j+1}<\ldots<K_N$,
where $\lK,\uK$ are additional strikes which can be chosen arbitrary
close to the barriers $\lb$ and $\ub$ respectively.
\end{lemma}

\begin{proof} 
The existence of a model-free arbitrage if any of
\eqref{eq:Cup}--\eqref{eq:Dbounds} do not hold is straightforward,
and can generally be read off from the inequality --- \eg{} if
\eqref{eq:Clow2} fails at $\lb$, the arbitrage is to buy the call
with strike $C(\lb)$, sell a call with strike $K_{i+1}$ and buy
$\frac{K_{i+1}-\lb}{K_{i+2}-K_{i+1}}$ units of the call with strike
$K_{i+2}$, and sell the same number of calls with strike $K_{i+1}$.

So suppose \eqref{eq:Cup}--\eqref{eq:Dbounds} hold and extend
$\mathcal{P}$ to $\mathcal{X}_D$ with prices as given in the
statement of the lemma. Assume that neither of \eqref{eq:wa1} or
\eqref{eq:wa2} hold. Then it is easy to check that we can find a piecewise linear
extension of the function $C(\cdot)$ which passes through each of the
specified call prices and satisfies \eqref{eq:Cbasic}--\eqref{eq:Ccnv} and \eqref{eq:D_basic}. Further, $C(\cdot)$ may be taken as in the statement of the Lemma, with $\lK,\uK$ arbitrary close to the barriers $\lb,\ub$ respectively. Proposition \ref{prop:wflvr} then grants
existence of a $(\mathcal{P}, \mathcal{X_D})$-market model. In particular there is no weak arbitrage and hence no model-free arbitrage.

It remains to argue that if either of \eqref{eq:wa1} or \eqref{eq:wa2} hold then there is a weak-arbitrage, but no model-free arbitrage.

We first rule out a model-free arbitrage. Since, by Proposition~\ref{prop:calls_WA}, $\mathcal{P}$ on
 $\mathcal{X}\cup \{(S_T-\lb)^+, (S_T-\ub)^+\}$ admits no WA we must only show that the
addition of the digital options does not introduce an
arbitrage. Consider initially the case where we add simply a digital
option at $\lb$. Recalling \eqref{eq:Pnoarb}, it is therefore
sufficient to show that there does not exist an $H \in
\Lin(\mathcal{X}\cup\{(S_T-\lb)^+,(S_T-\ub)^+\})$, such that either
\begin{align}
 \indic{S_T > \lb} + H \ge 0, &\quad \mbox{ and } \mathcal{P}H <
 -\iD(\lb)\label{eq:barrierarb1}\\
 \intertext{or}
 -\indic{S_T > \lb} + H \ge 0, &\quad \mbox{ and } \mathcal{P}H <
 \iD(\lb).\label{eq:barrierarb2}
\end{align}
However, any possible choice of $H$ must be the combination of a sum
of call options and forward options. Since only finitely many terms
may be included, we necessarily have $H$ as a continuous function in
$S_T$: more precisely, for $(S_t) \in \pathsp$, and any $\eps, \delta
0$, there exist paths $(S_t') \in \pathsp$ such that $S_T' \in (S_T
-\delta,S_T)$ and $|H((S_t)) -H((S_t'))| < \eps$. Clearly, we may
also replace the first conclusion with $S_T' \in (S_T,S_T +\delta)$.

So suppose \eqref{eq:barrierarb2} holds, and $H((S_t)) < 1$ for some
$(S_t) \in \pathsp$ with $S_T = \lb$, then by the continuity property
there exist $(S_t')\in\pathsp$ such that $S_T' > \lb$ and $H((S_t'))
< 1$, contradicting \eqref{eq:barrierarb2}. Hence, we can conclude
that if an arbitrage exists for $\indic{S_T > \lb}$, the same
portfolio $H$ is an arbitrage for $\indic{S_T \ge \lb}$ if $\mathcal{P}\indic{S_T \ge \lb} = \iD(\lb)$. However, it is possible
to construct an admissible curve $C(\cdot)$ such that $C(\cdot)$ matches the
given prices at the $K_i$'s and $\lb$, and has
\[
-C_{-}'(\lb) =-\iD(\lb)\in \left( \frac{C(\lb) - C(K_{i+1})}{K_{i+1}-\lb},
 \frac{C(K_{i})- C(\lb)}{\lb-K_{i}}\right],
\]
and hence by arguments of Proposition \ref{prop:wflvr}, there exists a model
under which the calls and the digital option $\indic{S_T \ge \lb}$ are fairly priced, and so there cannot be an arbitrage. More
generally, if we wish to consider adding both digital options, a similar argument will work, the only aspect we need to be careful
about is that we can find a suitable extension to the call price function, however this is guaranteed by assumption {\bf (A)}.

Finally, we construct a weak arbitrage if either of
\eqref{eq:wa1} or \eqref{eq:wa2} hold. Suppose \eqref{eq:wa1} holds. Let $\Pr$ be a model, then
\begin{enumerate}
\item if $\Pr(S_T \in (\lb,K_{i+1})) =0$ we sell short the digital option, buy $\frac{1}{K_{i+1}-\lb}$ units of the call with strike $\lb$, and sell the same number of units of the call with
 strike $K_{i+1}$. By \eqref{eq:wa1} we receive cash
 initially, but the final payoff is strictly negative only for
 $S_T \in (\lb,K_{i+1})$,
\item else $\Pr(S_T \in (\lb, K_{i+1})) >0$ and we purchase $\frac{1}{K_{i+1}-\lb}$ units of the call with strike $\lb$, sell
 $\frac{1}{K_{i+1} - \lb}+\frac{1}{K_{i+2} - K_{i+1}}$ units of the call
 with strike $K_{i+1}$ and buy $\frac{1}{K_{i+2} - K_{i+1}}$ units of the call with
 strike $K_{i+2}$. By \eqref{eq:wa1} this costs nothing initially, has non-negative payoff which is further strictly positive payoff in $(\lb,K_{i+1})$.
\end{enumerate}
A similar argument gives weak arbitrage when \eqref{eq:wa2} holds.
\end{proof}
\begin{cor} \label{cor:wkarb} {Suppose $\mathcal{P}$ on
  $\mathcal{X}_D$ does not admit a model-free arbitrage. Then for
  all $\eps > 0$, there exists a pricing operator $\mathcal{P}_\eps$
  on $\mathcal{X}_D$ such that $|\mathcal{P}X - \mathcal{P}_\eps X|
  < \eps$ for all $X \in \mathcal{X}_D$, and such that a
  $(\mathcal{P}_\eps, \mathcal{X}_D)$-market model exists. }
\end{cor}

\begin{proof}
Given the prices of the call and digital options, we can perturb
these prices by some small $\delta >0$ to get prices which satisfy
the no-weak-arbitrage conditions of Lemma~\ref{lem:wkarb}: more
precisely, if say \eqref{eq:wa1} holds, by taking $\tilde{C}(\lb) =
C(\lb) + \delta$, and perhaps (if needed to preserve
\eqref{eq:Dbounds}) also $\tilde{\iD}(\lb)= \iD(\lb) +
\frac{\delta}{K_{i+1}-\lb}$, and if necessary, performing a similar
operation at $\ub$. The corresponding pricing operator then
satisfies the stated conditions.
\end{proof}

\begin{proof}[Proof of Theorem~\ref{thm:ubfinstrike}]
Note that we can't apply directly Lemma \ref{lem:mf_bounds} to deduce \eqref{eq:ub2B} as we do not have prices of digital calls. Instead, we will devise superhedges of $\uh^I,\uh^{II}(K)$ and $\uh^{III}(K)$ which only use traded options. Recall that $i,j$ are such that $K_i<\lb<K_{i+1}$ and $K_j<\ub<K_{j+1}$. Consider the following payoffs, presented in Figure \ref{fig:binsh12},
\begin{equation}\label{eq:newsuperhedges}
\begin{split}
X_1 = & \frac{(S_T-K_{i-1})^+ - (S_T-K_i)^+}{K_i-K_{i-1}},\\
X_2  = & \frac{1}{\lb-K_i}\left((S_T-K_i)^+ - \frac{K_{i+2}-\lb}{K_{i+2}-K_{i+1}}(S_T-K_{i+1})^++\frac{K_{i+1}-\lb}{K_{i+2}-K_{i+1}}(S_T-K_{i+2})^+\right),
\end{split}
\end{equation}
which both superhedge the digital call $\indic{S_T>\lb}$. Recall that $C(\cdot), D(\cdot)$ at $\lb,\ub$ are now \emph{defined} via \eqref{eq:Clbdefn}--\eqref{eq:dubdefn}. Direct computation shows that if in \eqref{eq:Clbdefn} we have
\[
C(\lb) = C(K_i) + \frac{C(K_i)-C(K_{i-1})}{K_i-K_{i-1}} (\lb-K_i),
\]
and \eqref{eq:dlbdefn} holds then
\[
\Pc(X_1)=\frac{C(K_{i-1}) - C(K_i)}{K_i-K_{i-1}}=\iD(\lb).
\]
Likewise, if $C(\lb)$ is equal to the other term in \eqref{eq:Cubdefn}, using \eqref{eq:dlbdefn} we deduce that $\Pc(X_2)=\iD(\lb)$. Using similar ideas we construct superhedges $\uh^{II}_1(K)$, $i=1,2$ of $\uh^{II}(K)$ which only use traded options. They are graphically presented in Figure \ref{fig:dntsh12}. Further, just as above, it follows easily that
\begin{equation*}
\min\left\{\Pc\uh_1^{II}(K),\Pc\uh_{2}^{II}(K)\right\}=\iD(\lb)+\frac{C(K)-P(\lb)}{K-\lb},\quad K\in\K\cap (\lb,\infty).
\end{equation*}
In a similar manner we construct superhedges $\uh^{III}_i(K)$, $i=1,2$, of $\uh^{III}(K)$ which only use traded options and such that 
\begin{equation*}
\min\left\{\Pc\uh_1^{III}(K),\Pc\uh_{2}^{III}(K)\right\}=(1-\sD(\ub))+\frac{P(K)-C(\lb)}{K-\lb},\quad K\in\K\cap (\lb,\infty).
\end{equation*}
Finally we consider $\uh^{I}$, i.e.\ we need to construct a superhedge for $\indic{\lb < S_T < \ub}$. This can be done using a combination of the techniques used above for superhedging the digital calls. We end up with $\uh^I_i$, $i=1,2,3,4$ and again
\[
\min\left\{\Pc\uh^I_i:i=1,2,3,4\right\}=\iD(\lb)-\sD(\ub).
\]
Under no model-free arbitrage on $\mathcal{X}_D\cup\{\dnt\}$ the price of the double no-touch has to be less than or equal to the prices of all the superhedges given above, which is less than or equal to the value in \eqref{eq:ub2B} (the former takes its minimum over a smaller range of strikes).
\psfrag{K-1}{$K_{i-1}$}
\psfrag{K}{$K_{i}$}
\psfrag{K1}{$K_{i+1}$}
\psfrag{K2}{$K_{i+2}$}
\psfrag{K3}{$K$}
\psfrag{lb}{$\lb$}
\begin{figure}[t]
 \begin{minipage}[t]{7cm}
 \hspace*{-1cm}  \includegraphics[width=7cm]{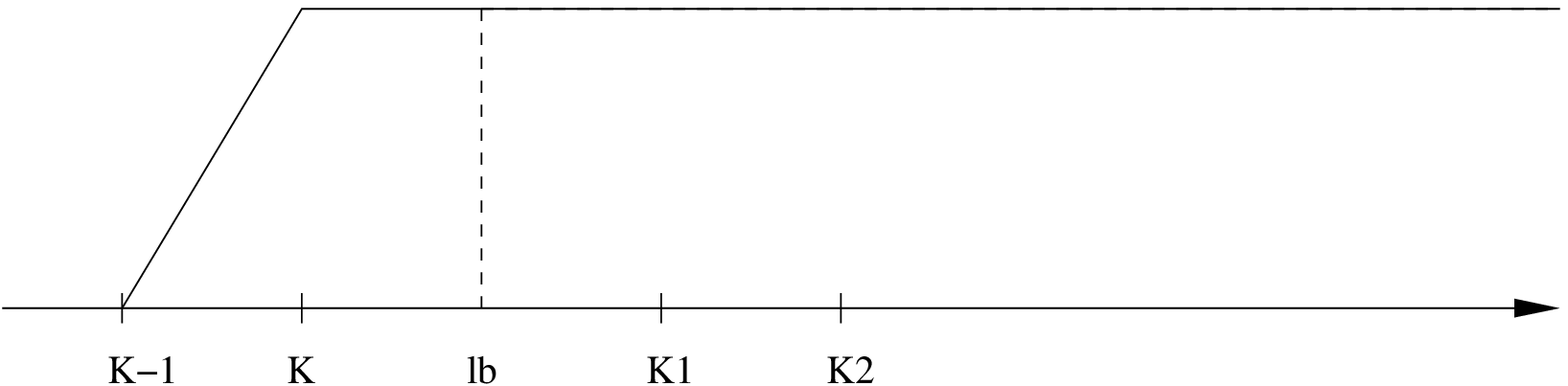}
   \begin{center}
    \hspace*{-1cm} (a) Type 1 digital call superhedge $X_1$.
   \end{center}
 \end{minipage}
 \hfill
 \begin{minipage}[t]{7cm}
\hspace*{-1cm}    \includegraphics[width=7cm]{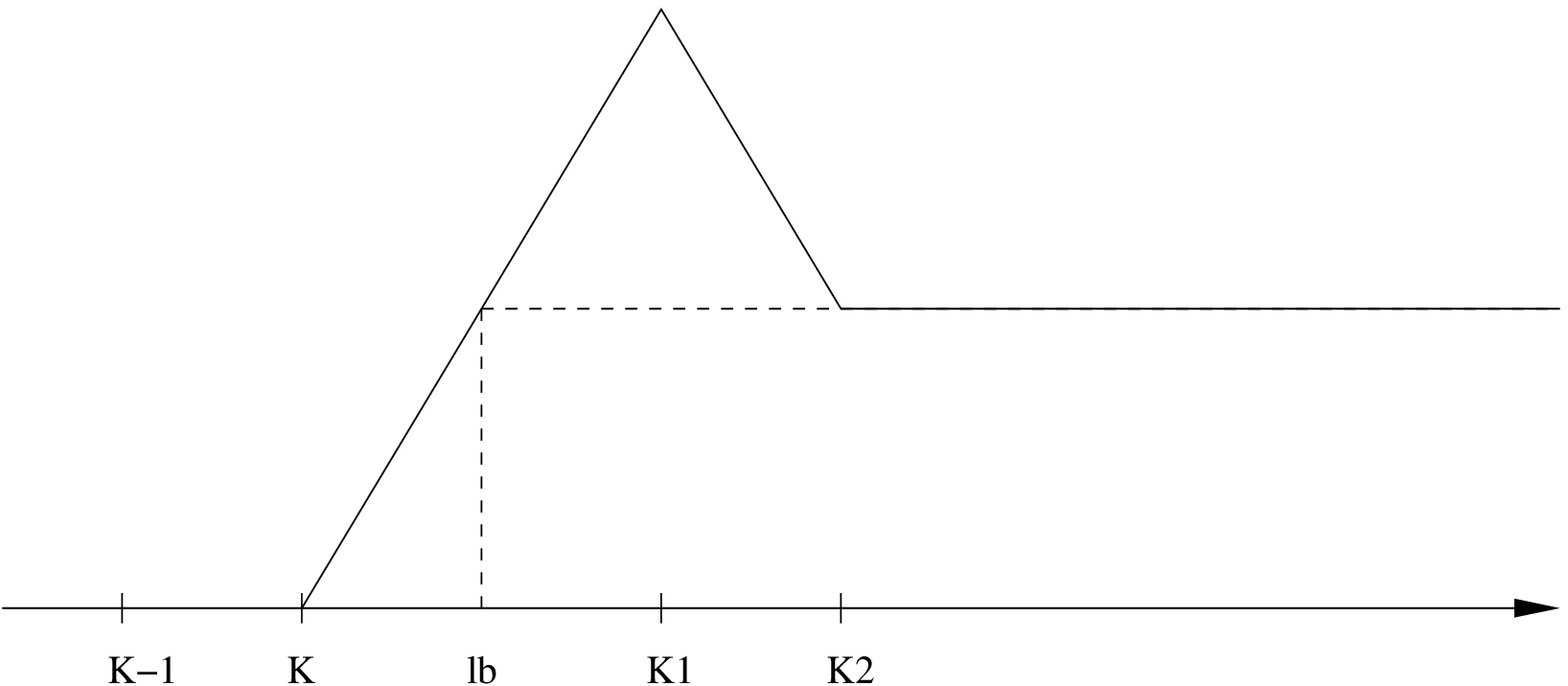}
   \begin{center}
    \hspace*{-1cm} (b) Type 2 digital call superhedge $X_2$.
   \end{center}
 \end{minipage}
 \caption{\label{fig:binsh12} The two possible superhedges $X_1,X_2$ of a digital call when
   calls and digital calls are not traded at the strike.}
\end{figure}



\begin{figure}[t]
 \begin{minipage}[t]{7cm}
  \hspace*{-1cm} \includegraphics[width=7cm]{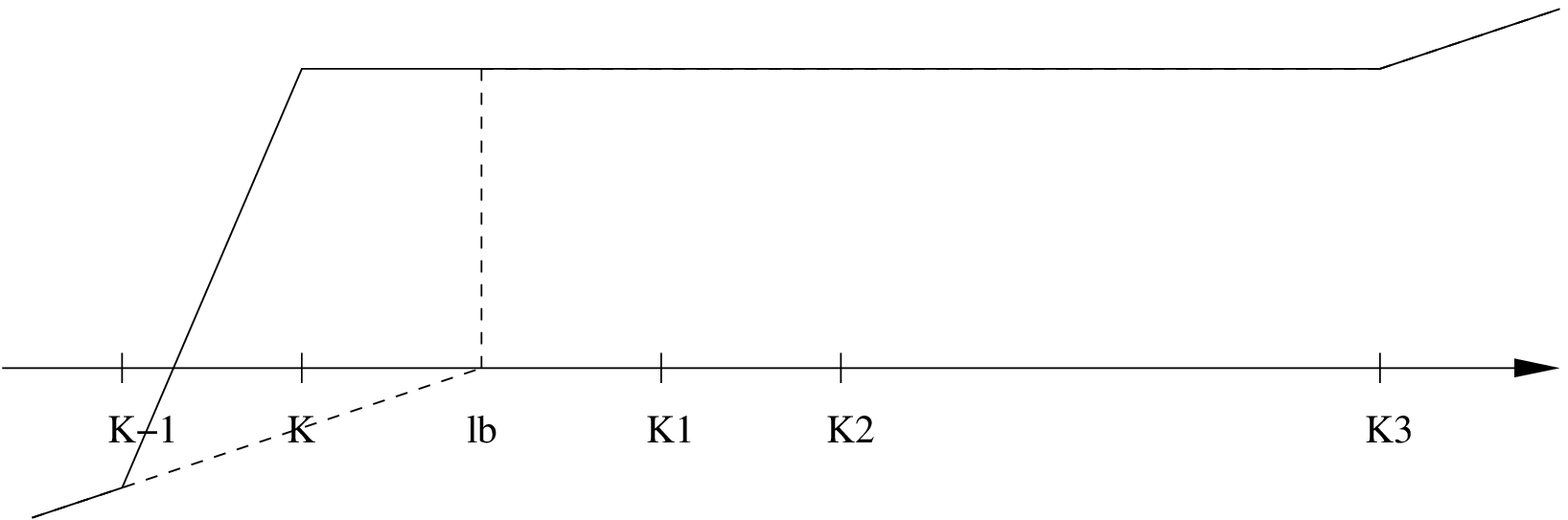}
   \begin{center}
  \hspace*{-1cm}   (a) Type 1 double no-touch superhedge.
   \end{center}
 \end{minipage}
 \hfill
 \begin{minipage}[t]{7cm}
  \hspace*{-1cm} \includegraphics[width=7cm]{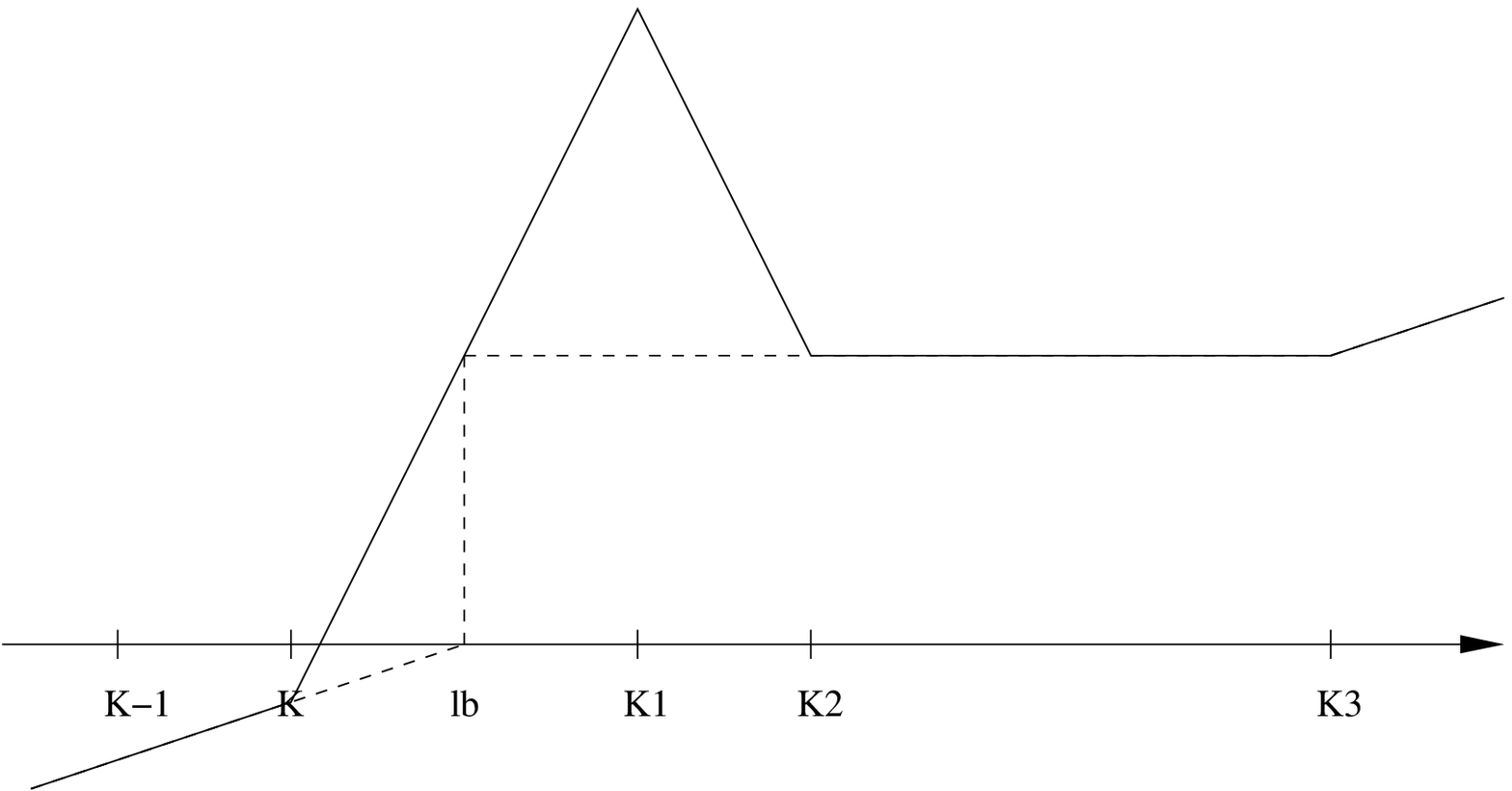}
   \begin{center}
   \hspace*{-1cm}  (b) Type 2 double no-touch superhedge.
   \end{center}
 \end{minipage}
 \caption{\label{fig:dntsh12} Two possible superhedges $\uh^{II}_1(K), \uh^{II}_2(K)$ of the
   double no-touch when calls and digitals at the barrier are not
   traded.}
\end{figure}
It remains to see that \eqref{eq:ub2B} is tight. Note that it follows from Theorem \ref{thm:ubfinstrike1} that \eqref{eq:ub2} remains unchanged if we extend the range of strikes considered therein to $K\in\K\cap(S_0,\infty)$ and $K\in \K\cap[0,S_0)$ in the second and third terms respectively, which we assume from now on.\\
Let us write $\tilde{\mathcal{X}}=\mathcal{X}(\K\cup\{\lb,\ub\})$ for the set defined via \eqref{eq:Xsetofassets} and \eqref{eq:stset_assume} but with $\K\cup\{\lb,\ub\}$ instead of $\K$. Note that by Lemma \ref{lem:wkarb},  $\Pc$ extended to $\tilde{\mathcal{X}}_D$ via \eqref{eq:Clbdefn}--\eqref{eq:dubdefn} admits no model-free arbitrage on $\tilde{\mathcal{X}}_D$ and we also know exactly when it admits a WA. Suppose that the minimum in \eqref{eq:ub2} does not occur in the second term for $K=\ub$ nor at the third term for $K=\lb$. Then \eqref{eq:ub2B} is equal to \eqref{eq:ub2}. If $\Pc$ admits no WA on $\tilde{\mathcal{X}}_D$ then the statements of the theorem follow directly from Theorem \ref{thm:ubfinstrike1}. If $\Pc$ admits a WA then, by Corollary \ref{cor:wkarb}, we can $\epsilon$-perturb $\Pc$ to $\tilde{\Pc}$ and have no WA on $\tilde{\mathcal{X}}_D$ and further, taking $\epsilon$ small enough, the minimum in \eqref{eq:ub2} for prices $\Pc$ and $\tilde{\Pc}$ occur at the same term and can be arbitrary close. The existence of a suitable sequence of $(\Pc,\mathcal{X})$ market models which approximate \eqref{eq:ub2B} follows as it suffices to take the $(\tilde{\Pc}, \tilde{\mathcal{X}}_D)$-market models which achieve the upper bound \eqref{eq:ub2}.

It remains to consider the case when the (strict) minimum in \eqref{eq:ub2}, for prices $\Pc$, occurs either at the second term for $K=\ub$ or at the third term for $K=\lb$. We consider the former, the latter is analogous. We now have a difference between \eqref{eq:ub2B} and \eqref{eq:ub2}, namely  
\[
\iD(\lb) + \frac{C(\ub)-P(\lb)}{\ub-\lb}
\]
is strictly smaller than the upper bound in \eqref{eq:ub2B}. 
Consider a new value of the call at $\ub$, $C^*(\ub)$, and the corresponding digital call price
\begin{equation}\label{eq:Dstardefn}
\sD^*(\ub) = \frac{C^*(\ub) - C(K_{j+1})}{K_{j+1}-\ub}.
\end{equation}

Then as we increase $C^*(\ub)$ from the value of $C(\ub)$ given by
\eqref{eq:Cubdefn} (which is the smallest value consistent with no
arbitrage), we note that $\sD^*(\ub)$ increases, and in
\eqref{eq:ub2} (applied to the new set of prices), we see that the
first and third terms decrease (continuously), while the second term
increases (continuously). We will show that we can choose $C^*(\ub)$ in such a way that \eqref{eq:ub2} and \eqref{eq:ub2B} coincide. Namely, we take  
\begin{equation} \label{eq:Cstardefn} C^*(\ub) =
 \min\left\{C(K_{j+1}) - \frac{C(K_{j+1}) - P(\lb)}{K_{j+1}-\lb}
   (K_{j+1}-\ub), C(K_j) + (\ub-K_j) \frac{C(K_j) - P(\lb)}{K_j -
     \lb} \right\}
\end{equation}
and let $\Pc^*$ denote the pricing operator extended from $\mathcal{X}$ to $\tilde{\mathcal{X}}_D$ via \eqref{eq:Cstardefn}, \eqref{eq:Dstardefn} and \eqref{eq:Clbdefn}--\eqref{eq:dlbdefn}). We will show that
\begin{enumerate}
\item $\Pc^*$ admits no model-free arbitrage on $\tilde{\mathcal{X}}_D$,
\item the infimum in the second term in \eqref{eq:ub2}, for prices $\Pc^*$, is still attained at $\ub$, and
 this agrees now with the minimum for \eqref{eq:ub2B}:
 \begin{equation} \label{eqn:infeqCstar}
 \inf_{K \in \K \cap (S_0,\infty)} \frac{C(K)-P(\lb)}{K-\lb} =
 \frac{C^*(\ub) - P(\lb)}{\ub-\lb};
 \end{equation}
\item the second term in \eqref{eq:ub2} is still the smallest:
 \begin{eqnarray}
   \frac{C^*(\ub) - P(\lb)}{\ub-\lb} & \le &  - \sD^*(\ub) \leq -\sD(\ub)\label{eq:Cstar1}\\
   \iD(\lb) & \le & 1 + \inf_{K \in \K \cap [\lb,\ub) \cup \{\lb\}}
   \frac{P(K)-C^*(\ub)}{\ub-K}. \label{eq:Cstar2}
 \end{eqnarray}
\end{enumerate}
If all these conditions hold, and in (i) there is no weak arbitrage, then \eqref{eq:ub2} for $\Pc^*$ and \eqref{eq:ub2B} are equal and applying Theorem \ref{thm:ubfinstrike1}  we get a $(\Pc,\mathcal{X})$-market model which achieves the upper bound in \eqref{eq:ub2B}. If there is a weak arbitrage in (i) we first use Corollary \ref{cor:wkarb} and obtain a sequence of models which approximate \eqref{eq:ub2B}. The proof is then complete.
So we turn to proving the above statements. For (i), we simply need to show that $C^*(\ub)$ satisfies \eqref{eq:Cup}--\eqref{eq:Clow2}, i.e.
\[
C(\ub) \leq C^*(\ub) \le \frac{K_{j+1}-\ub}{K_{j+1}-K_j} C(K_{j}) +
\frac{\ub - K_j}{K_{j+1}-K_j} C(K_{j+1}).
\]
The first inequality follows from the bounds:
\begin{equation} \label{eq:CPrange}
 \frac{C(\ub) - P(\lb)}{\ub-\lb} \in \left[
   \frac{C(\ub)-C(K_j)}{\ub-K_j}, \frac{C(K_{j+1}) -
     C(\ub)}{K_{j+1}-\ub} \right]
\end{equation}
which in turn follow from the optimality of $\ub$ in \eqref{eq:ub2}, and the
convexity of the prices $C(\cdot)$. The upper bound can be checked by taking a suitable weighted average of the terms in \eqref{eq:Cstardefn}, where the weights are chosen so that the $P(\lb)$ term drops out. 
To show (ii), we note that, by definition of $C^*(\ub)$, we
have
\[
\frac{C^*(\ub) - P(\lb)}{\ub-\lb} = \min \left\{
 \frac{C(K_j)-P(\lb)}{K_j-\lb}, \frac{C(K_{j+1}) -
   P(\lb)}{K_{j+1}-\lb} \right\}
\]
and \eqref{eqn:infeqCstar} follows from convexity of $C$ and the fact that minimum in the second term of \eqref{eq:ub2} for prices $\Pc$, was attained at $\ub$. 
For (iii), \eqref{eq:Cstar1} follows from the definition of $\sD^*(\ub)$ and \eqref{eqn:infeqCstar}, whose value is negative. To deduce
\eqref{eq:Cstar2}, we note the following: suppose \eqref{eq:Cstar2}
fails, so that (rearranging)
\begin{equation} \label{eq:gradcond} 1 - \iD(\lb) < \sup_{K \in \K \cap
   [\lb,\ub) \cup \{\lb\}} \frac{C^*(\ub)-P(K)}{\ub-K}.
\end{equation}
In particular, taking $K=\lb$ on the RHS we deduce
\[
\frac{C^*(\ub) - P(\lb)}{\ub-\lb} > 1-\iD(\lb).
\]
which combined with \eqref{eq:Cstar1} yields $1 - \iD(\lb) < -\sD^*(\ub)$ which is clearly a contradiction as the first term is non-negative and the second non-positive.
\end{proof}

\section{On the joint law of the maximum and minimum of a UI martingale}
\label{ap:mart}
Let $(M_t:t\leq \infty)$ be a uniformly integrable continuous martingale starting at a
deterministic value $M_0\in\R$ and write $\mu$ for its terminal law,
$\mu\sim M_\infty$, where we assume $M$ is non trivial, i.e.\ $\mu\neq
\delta_{M_0}$. We denote by $-\infty\leq a_\mu<b_\mu\leq \infty$ the
bounds of the support of $\mu$, i.e.\ $[a_\mu,b_\mu]$ is the smallest
interval with $\mu([a_\mu,b_\mu])=1$. In this final appendix, we include some remarks on the joint law of the maximum
and minimum of uniformly integrable martingales, namely we study the function 
$p(\lb,\ub):=\Pr\big(\iM_\infty>\lb\textrm{ and }\sM_\infty <\ub \big)$ for $\lb\leq M_0\leq \ub$. Naturally, this is closely related what we have done so far: $(M_t)$ can be interpreted (at least for $a_\mu\geq 0$) as a market model $M_{\frac{u}{T-u}}=S_{u}$, $u\leq T$, with maturity $T$ call prices given via \eqref{eq:mudef}, and $p(\lb,\ub)$ is then the price of a double no-touch option.
\begin{prop}\label{prop:M_bounds}
We have the following properties:
\begin{enumerate}
\item $p(M_0,\ub)=p(\lb,M_0)=0$,
\item $p(\lb,\ub)=1$ on $[-\infty,a_\mu)\times(b_\mu,\infty]$,
\item $p$ is non-increasing in $\lb\in (a_\mu,M_0)$ and
 non-decreasing in $\ub\in (M_0,b_\mu)$,
\item for $a_\mu\leq \lb<M_0<\ub\leq\infty$ we have
 \begin{equation}\label{eq:uimart_bound}
   \Pr\big(\underline{B}_{\tau_J}>\lb\textrm{ and
   }\overline{B}_{\tau_J} <\ub \big)\leq p(\lb,\ub)\leq
   \Pr\big(\underline{B}_{\tau_P}>\lb\textrm{ and
   }\overline{B}_{\tau_P} <\ub \big),
 \end{equation}
 where $(B_t)$ is a standard Brownian motion with $B_0=M_0$,
 $\tau_P$ is the Perkins stopping time \eqref{eq:perkins} embedding
 $\mu$ and $\tau_J$ is the Jacka stopping time \eqref{eq:jacka},
 for barriers $(\lb,\ub)$, embedding $\mu$. 
\end{enumerate}
\end{prop}
The first three assertions of the proposition are clear. Assertion
$(iv)$ is a reformulation of Lemmas \ref{lem:perkins} and
\ref{lem:jacka}. It suffices to note that $(B_{t\wedge\tau_J})$,
$(B_{t\wedge\tau_P})$, $(M_t)$ are all UI martingales starting at
$M_0$ and with the same terminal law $\mu$ for $t=\infty$.

We think of $p(\cdot,\cdot)$ as a surface defined over the
quarter-plane $[-\infty,M_0]\times [M_0,\infty]$. Proposition
\ref{prop:M_bounds} describes boundary values of the surface,
monotonicity properties and gives an upper and a lower bound on the
surface. However we note that there is a substantial difference
between the bounds linked to the fact that $\tau_P$ does not depend on
$(\lb,\ub)$ while $\tau_J$ does. In consequence, the upper bound is
attainable: there is a martingale $(M_t)$, namely
$M_t=(B_{t\wedge\tau_P})$, for which $p$ is equal to the upper bound
for all $(\lb,\ub)$. In contrast a martingale $(M_t)$ for which $p$
would be equal to the lower bound does not exist. For the martingale
$M_t=(B_{t\wedge\tau_J})$, where $\tau_J$ is defined for some pair
$(\lb,\ub)$, $p$ will attain the lower bound in some neighbourhood of
$(\lb,\ub)$ which will be strictly contained in $(a_\mu,M_0)\times
(M_0,b_\mu)$. 

We end this appendix with a result which provides some further insight
into the structure of the bounds discussed above. In particular, we
can show some finer properties of the function $p(\lb,\ub)$ and its
upper and lower bounds.

\begin{theorem} \label{thm:boundstructure}
The function $p(\lb,\ub)$ is c\`agl\`ad in $\ub$ and c\`adl\`ag in
$\lb$. Moreover, if $p$ is discontinuous at $(\lb,\ub)$, then $\mu$
must have an atom at one of $\lb$ or $\ub$. Further:
\begin{enumerate}
\item if there is a discontinuity at $(\lb,\ub)$ of the form:
 \[
 \lim_{w_n \downarrow \ub} p(\lb,w_n) > p(\lb,\ub)
 \]
 then the function $g$ defined by
 \[
 g(u) = \lim_{w_n \downarrow \ub} p(u,w_n) - p(u,\ub), \qquad u \le
 \lb
 \]
 is non-increasing.
\item
 if there is a discontinuity at $(\lb,\ub)$ of the form:
 \[
 \lim_{u_n \uparrow \lb} p(u_n,\ub) > p(\lb,\ub)
 \]
 then the function $h$ defined by
 \[
 h(w) = \lim_{u_n \uparrow \lb} p(u_n,w) - p(\lb,w), \qquad w \ge
 \ub
 \]
 is non-decreasing.
\end{enumerate}
And, at any discontinuity, we will be in at least one of the
above cases.

Finally, we note that the lower bound (corresponding to the tilted-Jacka
construction) is continuous in $(a_\mu,M_0)\times (M_0,b_\mu)$, and
continuous at the boundary ($\ub = b_\mu$ and $\lb = a_\mu$) unless
there is an atom of $\mu$ at either $b_\mu$ or $a_\mu$, while the
upper bound (which corresponds to the Perkins construction) has
a discontinuity corresponding to every atom of $\mu$
\end{theorem}

Before we prove the above result, we note the following useful result,
which is a simple consequence of the martingale property:

\begin{prop} \label{prop:supatom}
Suppose that $(M_t)_{t \ge 0}$ is a UI martingale with $M_\infty
\sim \mu$. Then $\Pr(\sM_{\infty} = \ub) >0$ implies $\mu(\{\ub\})
\ge \Pr(\sM_{\infty} = \ub)$ and
\[
\{ \sM_{\infty} = \ub \} = \{ M_t = \ub, \ \forall t \ge H_{\ub}\}
\subseteq \{ M_\infty = \ub\} \quad \as{}.
\]
\end{prop}

\begin{proof}[Proof of Theorem~\ref{thm:boundstructure}]
We begin by noting that by definition of $p(\lb,\ub)$, we
necessarily have the claimed continuity/limiting
properties. Further,
\[
\liminf_{(u_n,w_n) \to (u,w)} p(u_n,w_n) \ge \Pr(\iM_\infty > \lb \mbox{
 and } \sM_\infty < \ub)
\]
and
\[
\limsup_{(u_n,w_n) \to (u,w)} p(u_n,w_n) \le \Pr(\iM_\infty \ge \lb \mbox{
 and } \sM_\infty \le \ub).
\]
It follows that the function $p$ is continuous at $(\lb,\ub)$ if
$\Pr(\iM_\infty = \lb) = \Pr(\sM_\infty = \ub) = 0$. By
Proposition~\ref{prop:supatom}, this is true when $\mu(\{\ub,\lb\})
= 0$.

Note that we can now see that at a discontinuity of $p$, we must be
in at least one of the cases (i) or (ii). This is because
discontinuity at $(\lb,\ub)$ is equivalent to
\[
\Pr(\iM_\infty \ge \lb \mbox{ and } \sM_\infty \le \ub) >
\Pr(\iM_\infty > \lb \mbox{ and } \sM_\infty < \ub),
\]
from which we can deduce that at least one of the events
\[
\{\iM_\infty > \lb \mbox{ and } \sM_\infty = \ub\}, \quad
\{\iM_\infty = \lb \mbox{ and } \sM_\infty < \ub\}, \quad
\{\iM_\infty = \lb \mbox{ and } \sM_\infty = \ub\}
\]
is assigned positive mass. However, by
Proposition~\ref{prop:supatom} the final event implies both
$M_\infty = \lb$ and $M_\infty = \ub$ which is
impossible. Consequently, at least one of the first two events must
be assigned positive mass, and these are precisely the cases (i) and
(ii).

Consider now case (i). We can rewrite the statement as: if $g(\lb) >
0$, then $g(u)$ is decreasing for $u<\lb$. Note however that
\begin{eqnarray*}
g(u) & = & \Pr(\iM_\infty > u \mbox{ and } \sM_\infty \le \ub) -
\Pr(\iM_\infty > u \mbox{ and } \sM_\infty < \ub) \\
& = & \Pr(\iM_\infty > u \mbox{ and } \sM_\infty = \ub)
\end{eqnarray*}
which is clearly non-increasing in $u$. {In fact, provided that
 $g(\lb)<\Pr(\sM_\infty=\ub)$, it follows from \eg{}
 \cite[Theorem~4.1]{Rogers:93} that $g$ is strictly decreasing for
 $\lb>u>\sup\{u\geq -\infty: g(u)=\Pr(\sM_\infty=\ub)\}$.  A similar
 proof holds in case (ii).}

We now consider the lower bounds corresponding to the tilted-Jacka
construction. We wish to show that
\[
\Pr(\iM_\infty \ge \lb \mbox{ and } \sM_\infty \le \ub) =
\Pr(\iM_\infty > \lb \mbox{ and } \sM_\infty < \ub),
\]
for any $(\ub,\lb)$ except those excluded in the statement of the
theorem. We note that it is sufficient to show that $\Pr(\iM_\infty
= \lb) = \Pr(\sM_\infty = \ub) = 0$, and by
Proposition~\ref{prop:supatom} it is only possible to have an atom
in the law of the maximum or the minimum if the process stops at the
maximum with positive probability; we note however that the stopping
time $\tau_J$, due to the definition of $\Psi_\mu$ and $\Theta_\mu$
precludes such behaviour except at the points $a_\mu, b_\mu$.

Considering now the Perkins construction, we note from
\eqref{eq:perkins} and the fact that the function $\gamma_\mu^+$ is
decreasing, that we will stop at $\lb$ only if $\gamma_\mu^+(\sM_t)
= \lb$ and $M_t = \iM_t = \lb$. It follows from \eqref{eq:gammaplus}
that there is a range of values $(\ub_*,\ub^*)$ for which
$\gamma_\mu^+(b) = \lb$, and consequently, we must have $h(b) =
\Pr(\iM_\infty = \lb, \sM_\infty < b)$ increasing in $b$ as $b$ goes
from $\ub_*$ to $\ub^*$, with $h(\ub_*) = \Pr(\iM_\infty = \lb,
\sM_\infty < \ub_*)=0$ and $h(\ub^*) = \Pr(\iM_\infty = \lb,
\sM_\infty < \ub^*)=\mu(\{\lb\})$.\footnote{In fact, as above, it follows from
 \eg{} \cite[Theorem~2.2]{Rogers:93} that the maximum must have a
 strictly positive density with respect to Lebesgue measure, and
 therefore that the function $h$ is strictly increasing between the
 points $\ub_*$ and $\ub^*$.} Similar results for the function $g$
also follow.
\end{proof}

Rogers \cite{Rogers:93} and Vallois \cite{Vallois:93} described all
possible joint laws of $(M_\infty,\sM_\infty)$. This naturally also
gives possible joint laws of $(M_\infty,\iM_\infty)$. It would be
interesting and natural to try to describe possible joint laws of the
triple $(M_\infty,\sM_\infty,\iM_\infty)$.
\smallskip\\
\textbf{Acknowledgement.} We are grateful to John Crosby whose
comments enhanced our understanding of FX markets and practicalities
of double no-touch options and who also provided us with some
numerical data from his own work.

\bibliographystyle{alpha}
\bibliography{general}

\end{document}